\documentclass[12pt, draftclsnofoot, onecolumn,romanappendices]{IEEEtran}

\usepackage{bm,cite}
\usepackage{amsmath}
\usepackage{amsthm}
\usepackage{amsbsy}
\usepackage{epsfig}
\usepackage{latexsym}
\usepackage{amssymb}
\usepackage{wasysym}
\usepackage{placeins}  
\usepackage{color}
\usepackage{rotating}  
\usepackage[usenames,dvipsnames]{xcolor}
\usepackage{dsfont}
\usepackage{caption}
\usepackage{subcaption}
\usepackage{algorithm}
\usepackage{algorithmic}
\usepackage{bbold}

\usepackage{tikz}

\usetikzlibrary{shapes,arrows,snakes}
\usetikzlibrary{calc}

\IEEEoverridecommandlockouts
\newtheorem{theorem}{Theorem}

\newtheorem{lemma}{Lemma}

\DeclareMathAlphabet{\mathbit}{OML}{cmr}{bx}{it}
\DeclareMathAlphabet{\mathsf}{OT1}{cmss}{m}{n}
\DeclareMathAlphabet{\mathTXf}{OT1}{cmss}{bx}{it}

\DeclareMathOperator{\trace}{tr}

\DeclareMathOperator{\rZF}{rZF}

\DeclareMathOperator{\TX}{TX}

\DeclareMathOperator{\DCSI}{DCSI}
\DeclareMathOperator{\Rate}{R} 
\DeclareMathOperator{\CN}{\mathcal{N}_{\mathbb{C}}}

\DeclareMathOperator{\SINR}{SINR}

\newcommand{\bA}{\mathbf{A}} 
 
\newcommand{\bC}{\mathbf{C}} 
 
\newcommand{\bE}{\mathbf{E}}

\newcommand{\bH}{\mathbf{H}} 
\newcommand{\bI}{\mathbf{I}}

\newcommand{\bL}{\mathbf{L}}

\newcommand{\bP}{\mathbf{P}} 
\newcommand{\bQ}{\mathbf{Q}} 
\newcommand{\bR}{\mathbf{R}} 
 
\newcommand{\bT}{\mathbf{T}} 
\newcommand{\bU}{\mathbf{U}} 
\newcommand{\bV}{\mathbf{V}}

\newcommand{\be}{\bm{e}}

\newcommand{\bh}{\bm{h}}

\newcommand{\bq}{\bm{q}} 
 
\newcommand{\bs}{\bm{s}} 
\newcommand{\bt}{\bm{t}}

\newcommand{\bz}{\bm{z}}

\newcommand{\LB}{\left(}
\newcommand{\RB}{\right)}
\newcommand{\LSB}{\left[}
\newcommand{\RSB}{\right]}

\newcommand{\I}{\mathbf{I}} %identity matrix
\newcommand{\norm}[1]{\lVert{#1}\rVert}

\newcommand{\Fro}{{\mathrm{F}}}

\newcommand{\E}{{\mathbb{E}}}

\newcommand{\trans}{{\text{T}}} 
 
\newcommand{\He}{{{\mathrm{H}}}}

\newcommand{\CCSI}{{\mathrm{CCSI}}}

\newcommand{\xv}{\mathbf{x}}
\newcommand{\yv}{\mathbf{y}}

\theoremstyle{remark}
\newtheorem{remark}{Remark} %[section]
\theoremstyle{example}

\begin{document} 
\title{Regularized ZF in Cooperative Broadcast Channels under Distributed CSIT: A Large System Analysis} 
\author{\IEEEauthorblockN{Paul de Kerret\IEEEauthorrefmark{1}, David Gesbert\IEEEauthorrefmark{1}, and Umer Salim\IEEEauthorrefmark{3}}

%%%%%%\IEEEauthorblockA{\IEEEauthorrefmark{2}
%%%%%%T\'el\'ecom Bretagne, IMT, UMR CNRS 3192 Lab-STICC} 

\IEEEauthorblockA{\IEEEauthorrefmark{1}
Mobile Communication Department, Eurecom}

\IEEEauthorblockA{\IEEEauthorrefmark{3}Intel Mobile Communications}
}

\maketitle
\begin{abstract}
Obtaining accurate Channel State Information (CSI) at the transmitters (TX) is critical to many cooperation schemes such as Network MIMO, Interference Alignment etc. Practical CSI feedback and limited backhaul-based sharing inevitably creates degradations of CSI which are \emph{specific} to each TX, giving rise to a \emph{distributed} form of CSI. In the Distributed CSI (D-CSI) broadcast channel setting, the various TXs design elements of the precoder based on their individual estimates of the global multiuser channel matrix, which intuitively degrades performance when compared with the commonly used centralized CSI assumption. This paper tackles this challenging scenario and presents a first analysis of the rate performance for the distributed CSI multi-TX broadcast channel setting, in the large number of antenna regime. Using Random Matrix Theory (RMT) tools, we derive deterministic equivalents of the Signal to Interference plus Noise Ratio (SINR) for the popular regularized Zero-Forcing (ZF) precoder, allowing to unveil the \emph{price of distributedness} for such cooperation methods.
\end{abstract}


%%%%%%%%%%%%%%%%%%%%%%%%%%%%%%%%%%%%%%%%%%%%%%
%%%%%%%%%%%%%%%%%%%%%%%%%%%%%%%%%%%%%%%%%%%%%%
%%%%%%%%%%%%%%%%%%%%%%%%%%%%%%%%%%%%%%%%%%%%%%
%%%%%%%%%%%%%%%%%%%%%%%%%%%%%%%%%%%%%%%%%%%%%%
\section{Introduction}
%%%%%%%%%%%%%%%%%%%%%%%%%%%%%%%%%%%%%%%%%%%%%%
%%%%%%%%%%%%%%%%%%%%%%%%%%%%%%%%%%%%%%%%%%%%%%
%%%%%%%%%%%%%%%%%%%%%%%%%%%%%%%%%%%%%%%%%%%%%%
%%%%%%%%%%%%%%%%%%%%%%%%%%%%%%%%%%%%%%%%%%%%%%

Network (or Multi-cell) MIMO methods, whereby multiple interfering TXs share user messages and allow for joint precoding, are currently considered for next generation wireless networks \cite{Gesbert2010}. With perfect message and CSI sharing, the different TXs can be seen as a unique virtual multiple-antenna array serving all RXs in a multiple-antenna broadcast channel (BC) fashion, and well known precoding algorithms from the literature can be used \cite{Karakayali2006}. Joint precoding however requires the feedback of an accurate multi-user CSI to each TX in order to achieve near optimal sum rate performance\cite{Jindal2006}.

Although the case of imperfect, noisy or delayed, CSI has been considered in past work \cite{Jindal2006,MaddahAli2012}, literature typically assumes \emph{centralized} CSIT, i.e., that the precoding is done on the basis of a \emph{single} imperfect channel estimate which is common at every TX. Although meaningful in the case of a broadcast with a single transmit device, this assumption can be challenged when the joint precoding is carried out across distant TXs linked by heterogeneous and imperfect backhaul links or having to communicate without backhaul (over the air) among each other, as in the case of direct device-to-device cooperation. In these cases, it is expected that the CSI exchange will introduce further delay and quantization noise. It is thus practically relevant for joint precoding across distant TXs to consider a CSI setting where each TX receives its \emph{own} channel estimate. This setting is referred to as distributed CSI (D-CSI) in the rest of this paper. % which we denote as the \emph{distributed CSI} configuration \cite{dekerret2013_thesis}.

From an information theoretic perspective, the study of transmitter cooperation in the D-CSI broadcast channel setting raises several intriguing and challenging questions.

First, the capacity region of the broadcast channel under a general D-CSI setting is unknown. In \cite{dekerret2012_TIT}, a rate characterization at high SNR is carried out using DoF analysis for the two transmitters scenario. This study highlighted the severe penalty caused by the lack of a consistent CSI shared by the cooperating TXs from a DoF point of view, when using a conventional precoder. It was also shown that classical (regularized) robust  precoders \cite{Peel2005} do not restore the DoF. Although a new DoF-restoring decentralized precoding strategy was presented in \cite{dekerret2012_TIT} for the two TXs case, the general case of more than $2$ TXs remains open. At finite SNR, the problem of designing precoders that optimally tackle the D-CSI setting is open for any number of TXs. The use of conventional linear precoders that are unaware of the D-CSI structure is expected to yield a loss with respect to a centralized (even imperfect) CSI setting. Yet the quantifying of this loss in the finite SNR region has not been addressed previously. This is precisely the question addressed by this paper.  
 
We study the average rate achieved when the number of transmit antennas and the number of receive antennas jointly grow large with a fixed ratio, thus allowing to use efficient tools from the field of RMT. Although RMT has been applied in many works to the analysis of wireless communications [See~\cite{Hochwald2002,Tulino2007,Wagner2012,Muller2013,Couillet2011} among others], its role in helping to analyze cooperative systems with distributed information has not been highlighted before.

Our main contribution consists in providing a deterministic equivalent for the average rate per user in a D-CSI setting where each TX receives its \emph{own} estimate of the global multi-user channel matrix with the quality (in a statistical sense) of this estimate varying from TX to TX.
%%%\item Building upon results from \cite{Wagner2012,Muller2013}, we have developed a new deterministic equivalent for a term involving two resolvent matrices related to two different channel estimates. 
%%\end{itemize}

%%%%%%%%%%%%%%%%%%%%%%%%%%%%%%%%%%%%%%%%%%%%%%
%%%%%%%%%%%%%%%%%%%%%%%%%%%%%%%%%%%%%%%%%%%%%%
%%%%%%%%%%%%%%%%%%%%%%%%%%%%%%%%%%%%%%%%%%%%%%
%%%%%%%%%%%%%%%%%%%%%%%%%%%%%%%%%%%%%%%%%%%%%%
\section{System Model}\label{se:SM}
%%%%%%%%%%%%%%%%%%%%%%%%%%%%%%%%%%%%%%%%%%%%%%
%%%%%%%%%%%%%%%%%%%%%%%%%%%%%%%%%%%%%%%%%%%%%%
%%%%%%%%%%%%%%%%%%%%%%%%%%%%%%%%%%%%%%%%%%%%%%
%%%%%%%%%%%%%%%%%%%%%%%%%%%%%%%%%%%%%%%%%%%%%%

\subsection{Transmission Model}
%%%%%%%%%%%%%%%%%%%%%%%%%%%%%%%%%%%%%%%%%%%%%%
%%%%%%%%%%%%%%%%%%%%%%%%%%%%%%%%%%%%%%%%%%%%%%
We study a so-called Network MIMO transmission where $n$~TXs serve jointly $K$~receivers (RXs). In order to simplify our analysis  we restrict ourselves to linear precoding structures. Each TX is equipped with $M_{\TX}$~antennas such that the total number of transmit antennas is denoted by $M\triangleq n M_{\TX}$. Every RX is equipped with a single-antenna. We assume that the ratio of transmit antennas with respect to the number of users is fixed and given by 
\begin{equation}
\beta\triangleq \frac{M}{K}\geq 1.
\end{equation}
%We consider that the scheduling step has already be done such that $\beta\geq 1$ in order to efficiently reduce interference.

We further assume that the RXs have perfect CSI so as to focus on the imperfect CSI feedback and exchange among the TXs. We consider that the RXs treat interference as noise. The channel from the $n$~TXs to the $K$~RXs is represented by the multi-user channel matrix~$\mathbf{H} \in \mathbb{C}^{K\times M}$. 

The transmission is then described as
\begin{equation}
\begin{bmatrix}
y_1\\\vdots\\y_K
\end{bmatrix}
=
\mathbf{H}\xv
+
\bm{\eta}
=
\begin{bmatrix}
\bh_1^{\He}\xv\\\vdots\\
\bh_K^{\He}\xv
\end{bmatrix}
+
\begin{bmatrix}
\eta_1\\\vdots\\
\eta_K
\end{bmatrix}
\label{eq:SM_2}
\end{equation}
where $y_i\in \mathbb{C}$ is the signal received at the $i$-th RX, $\bh^{\He}_i=\bm{e}_i^{\He}\bH \in\mathbb{C}^{1\times M}$ is the channel from all transmit antennas to RX~$i$, and $\bm{\eta}\triangleq[\eta_1,\ldots,\eta_K]^{\trans}\in \mathbb{C}^{K\times 1}$ is the normalized Gaussian noise with its elements i.i.d. as $\mathcal{CN}(0,1)$. 

The transmitted multi-user signal~$\xv\in \mathbb{C}^{M\times 1}$ is obtained from the symbol vector $\bs\triangleq[\bs_1^{\trans},\ldots,\bs_K^{\trans}]^{\trans}  \in  \mathbb{C}^{K\times 1}$ with its elements i.i.d. $\CN(0,1)$ as
\begin{equation}
\xv=\bT \bs=
\begin{bmatrix}
\bt_1,
\hdots,
\bt_K
\end{bmatrix}
\begin{bmatrix}
\bs_1\\
\vdots\\
\bs_K
\end{bmatrix}
\label{eq:SM_3}
\end{equation}   
with $\bT \in \mathbb{C}^{M\times K}$ being the \emph{multi-user} precoder and $\bt_i\triangleq\bT\be_i \in \mathbb{C}^{M\times 1}$ being the beamforming vector used to transmit to RX~$i$. We consider for clarity the sum power constraint~$\|\bT\|_{\Fro}^2=P$.

Our main figure-of-merit is the average rate per user
\begin{equation}
\Rate\triangleq\frac{1}{K}\sum_{k=1}^K\E\LSB \log_2\LB 1+\SINR_k\RB\RSB
\label{eq:SM_4}
\end{equation} 
where~$\SINR_k$ denotes the Signal-to-Interference and Noise Ratio (SINR) at RX~$k$ and is defined as
\begin{equation}
\SINR_k\triangleq\frac{\left|\bh_k^{\He}\bt_k\right|^2}{1+\sum_{\ell=1,\ell \neq k}^K\left|\bh_k^{\He}\bt_{\ell}\right|^2}.
\label{eq:SM_5}
\end{equation} 
%%%%%%%%%%%%%%%%%%%%%%%%%%%%%%%%%%%%%%%%%%%%%%
%%%%%%%%%%%%%%%%%%%%%%%%%%%%%%%%%%%%%%%%%%%%%%
\subsection{Distributed CSIT Model}
%%%%%%%%%%%%%%%%%%%%%%%%%%%%%%%%%%%%%%%%%%%%%%
%%%%%%%%%%%%%%%%%%%%%%%%%%%%%%%%%%%%%%%%%%%%%%
In the distributed CSIT model studied here, each TX receives its own CSI based on which it designs its transmission parameters without any additional communication to the other TXs. %The actual feedback and exchange mechanism based on which the TXs receive the multi-user channel estimate is completely arbitrary, yet it can itself be the topic of some interesting trade-off and optimization beyond the scope of this paper \cite{Kobayashi2011,Cheng2010}. 
Specifically, TX~$j$ receives the multi-user channel estimate~$\hat{\bH}^{(j)} \in \mathbb{C}^{K\times M}$ and designs its transmit coefficients~$\xv_j\in \mathbb{C}^{M_{\TX}\times 1}$ solely as a function of~$\hat{\bH}^{(j)}$. For ease of exposition, we assume that the imperfect multi-user channel estimate is modeled by
\begin{align}
\hat{\bH}^{(j)}\triangleq \sqrt{1-(\sigma^{(j)})^2}\bH+\sigma^{(j)}\bm{\Delta}^{(j)}
\label{eq:SM_6}
\end{align}
with~$\bm{\Delta}^{(j)}\in \mathbb{C}^{K\times M}$ having its elements i.i.d.~$\CN(0,1)$. We then denote by~$(\hat{\bh}_k^{(j)})^{\He}$ the $k$th row of $\hat{\bH}^{(j)}$, i.e., the estimate at TX~$j$ of the channel from all the transmit antennas to RX~$k$. The approach described in this work extends to a more general information structure allowing for a non uniform description quality at each TX as well as correlation in the multi-user channel. The calculations will be provided in an upcoming work.
\begin{remark}
The D-CSI model encompasses the imperfect centralized CSI model by taking~$n=1$.\qed
\end{remark}

%%%%%%%%%%%%%%%%%%%%%%%%%%%%%%%%%%%%%%%%%%%%%%
%%%%%%%%%%%%%%%%%%%%%%%%%%%%%%%%%%%%%%%%%%%%%%
\subsection{Regularized Zero Forcing with Distributed CSIT}\label{se:SM:ZF}
%%%%%%%%%%%%%%%%%%%%%%%%%%%%%%%%%%%%%%%%%%%%%%
%%%%%%%%%%%%%%%%%%%%%%%%%%%%%%%%%%%%%%%%%%%%%%
We address the performance of a classical MISO broadcast precoder, namely \emph{regularized ZF} \cite{Spencer2004,Peel2005}, when faced with distributed CSIT in the large system regime. Hence, the precoder designed at TX~$j$ is assumed to take the form 
\begin{align}
\bT_{\rZF}^{(j)}\triangleq  \LB \hat{\bH}^{(j)}(\hat{\bH}^{(j)})^{\He}+M\alpha \I_{M} \RB ^{-1} \hat{\bH}^{(j)} \frac{\sqrt{P}}{\sqrt{\Psi}}
\label{eq:SM_7}
\end{align}
with regularization factor $\alpha>0$. We also define
\begin{align}
\bC^{(j)}\triangleq  \frac{ \hat{\bH}^{(j)}(\hat{\bH}^{(j)})^{\He}}{M}+\alpha\I_{M} 
\label{eq:SM_8}
\end{align}
such that the precoder at TX~$j$ can be rewritten as
\begin{align}
\bT_{\rZF}^{(j)}=  \frac{1}{M}(\bC^{(j)})^{-1} \hat{\bH}^{(j)} \frac{\sqrt{P}}{\sqrt{\Psi^{(j)}}}.
\label{eq:SM_9}
\end{align}
The scalar~$\Psi^{(j)}$ corresponds to the power normalization at TX~$j$. Hence, it holds that
\begin{align}
\Psi^{(j)}&=\norm{\LB \hat{\bH}^{(j)}(\hat{\bH}^{(j)})^{\He}+M\alpha \I_{M} \RB ^{-1}\hat{\bH}^{(j)}}^2_{\Fro}.
\label{eq:SM_10}
\end{align}
%%%%%%\begin{remark}
%%%%%%Because of the distributed precoding, it is not clear a priori whether the normalization done at each TX leads to fulfill the power constraint. However, it will become clear that in the large antenna regime, these quantities become deterministic and the power constraint can then be fulfilled with equality. \qed
%%%%%%\end{remark}
Upon concatenation of all TX's precoding vectors, the effective global precoder denoted by $\bT_{\rZF}^{\DCSI}$, is equal to
\begin{equation}
\bT_{\rZF}^{\DCSI} \triangleq 
\begin{bmatrix} 
\bE_1^{\He}\bT_{\rZF}^{(1)}\\
\bE_2^{\He}\bT_{\rZF}^{(2)}\\
\vdots\\
\bE_K^{\He}\bT_{\rZF}^{(K)}
\end{bmatrix}
\label{eq:SM_10}
\end{equation} 
where $\bE_j^{\He}\in \mathbb{C}^{M_{\TX}\times M}$ is defined as
\begin{equation}
\bE_j^{\He}\triangleq 
\begin{bmatrix}
\bm{0}_{M_{\TX}\times (j-1)M_{\TX}}&\I_{M_{\TX}}&\bm{0}_{M_{\TX}\times (n-j)M_{\TX}}
\end{bmatrix}.
\label{eq:SM_11}
\end{equation} 
We furthermore denote the $k$th column of $\bT_{\rZF}^{\DCSI}$ (used to serve RX~$k$) by $\bt_{\rZF,k}^{\DCSI}$.

Although the finite SNR rate analysis under the precoding structure \eqref{eq:SM_10} and the distributed CSI model in \eqref{eq:SM_6} is challenging in the general case because of the dependency of one user performance on a all channel estimates, some useful insights can be obtained in the large antenna regime, as shown below.

%%%%%%%%%%%%%%%%%%%%%%%%%%%%%%%%%%%%%%%%%%%%%%%%%%%%%%%%%%%%%%%%%%%%%%%%%%%%%
%%%%%%%%%%%%%%%%%%%%%%%%%%%%%%%%%%%%%%%%%%%%%%%%%%%%%%%%%%%%%%%%%%%%%%%%%%%%%
%%%%%%%%%%%%%%%%%%%%%%%%%%%%%%%%%%%%%%%%%%%%%%%%%%%%%%%%%%%%%%%%%%%%%%%%%%%%%
%%%%%%%%%%%%%%%%%%%%%%%%%%%%%%%%%%%%%%%%%%%%%%%%%%%%%%%%%%%%%%%%%%%%%%%%%%%%%
\section{Deterministic Equivalent of the SINR}\label{se:main}
%%%%%%%%%%%%%%%%%%%%%%%%%%%%%%%%%%%%%%%%%%%%%%%%%%%%%%%%%%%%%%%%%%%%%%%%%%%%%
%%%%%%%%%%%%%%%%%%%%%%%%%%%%%%%%%%%%%%%%%%%%%%%%%%%%%%%%%%%%%%%%%%%%%%%%%%%%%
%%%%%%%%%%%%%%%%%%%%%%%%%%%%%%%%%%%%%%%%%%%%%%%%%%%%%%%%%%%%%%%%%%%%%%%%%%%%%
%%%%%%%%%%%%%%%%%%%%%%%%%%%%%%%%%%%%%%%%%%%%%%%%%%%%%%%%%%%%%%%%%%%%%%%%%%%%%
The Stieltjes transform has proven very useful many times in the analysis of wireless networks \cite{Wagner2012,Muller2013} and we will also follow this approach. Hence, our approach will be based on the following fundamental result.
\begin{theorem}{\cite{Hachem2010,Muller2013}}
\label{fundamental_theorem}
Consider the resolvent matrix $\bQ\triangleq\LB\frac{\bH^{\He}\bH}{M} +\alpha\I_M\RB^{-1}$ with the matrix~$\bH$ defined according to Section~\ref{se:SM} and $\alpha>0$. Then the equation
\begin{equation}
x=\frac{1}{M}\trace \LB \LB \alpha\bI_M+\frac{\I_M}{\beta \LB 1+x \RB }\RB^{-1}\RB 
\label{eq:main_1}
\end{equation}
admits a unique fixed point which we will denote by~$\delta$ in the following. Let 
\begin{equation}
\bQ_o\triangleq \LB \alpha\bI_M+\frac{\I_M}{\beta \LB 1+\delta\RB }\RB^{-1}
\label{eq:main_2}
\end{equation}
and let the matrix~$\bU$ be any matrix with bounded spectral norm. Then,
\begin{equation}
\frac{1}{M}\trace \LB \bU\bQ \RB -\frac{1}{M}\trace \LB \bU\bQ_o \RB    \xrightarrow[K,M\rightarrow \infty]{a.s.}0.
\label{eq:main_3}
\end{equation} 
\end{theorem}  
The fixed point~$\delta$ can easily be obtained by an iterative algorithm given in \cite{Wagner2012,Couillet2011} and recalled in Appendix~\ref{app:literature} for the sake of completeness. Using this theorem and the definition of~$\delta$, we can now state our main result.
\begin{theorem}
\label{theorem}
Considering the D-CSI BC described in Section~\ref{se:SM}, then $\SINR_k-\SINR_k^o\rightarrow 0$
with $\SINR_k^o$ defined as
\begin{align}
\SINR_k^o\triangleq \frac{\frac{1}{n}\sum_{j=1}^n  \sqrt{1-(\sigma^{(j)})^2} \frac{\delta}{1+\delta}}{I_k^o+\frac{\Gamma^o}{P}}
\label{eq:main_4}
\end{align} 
with~$I_k^o\in\mathbb{R}$ given by
\begin{align}
I_k^o&\triangleq\sum_{j=1}^n\frac{\Gamma^o}{(1+\delta)^2n^2}\LSB n+2\delta\LB -1+n+(\sigma^{(j)})^2\RB\notag \right.\\
&\left.+\delta^2\LB -1+n+(\sigma^{(j)})^2\RB+\delta^4\LB-(\sigma^{(j)})^6+(\sigma^{(j)})^8\RB\notag \right.\\
&\left.+(\sigma^{(j)})^5\sqrt{1-(\sigma^{(j)})^2}\sqrt{(\sigma^{(j)})^2-(\sigma^{(j)})^4} \RSB\notag        \\
&+\sum_{j=1}^n\sum_{j'=1,j'\neq j}^n\frac{\Gamma^o_{j,j'}\delta}{(1+\delta)^2n^2}\LSB-2+(\sigma^{(j)})^2+(\sigma^{(j')})^2\right.\notag \\
&\left.+\delta\LB-1+(\sigma^{(j)})^2+(\sigma^{(j')})^2 \RB \RSB       
\label{eq:main_5}
\end{align} 
while $\Gamma^o\in\mathbb{R}$ and $\Gamma^o_{j,j'}\in\mathbb{R}$ are respectively defined as
\begin{align}
\Gamma^o&\triangleq  \frac {\frac{  \delta^2}{\beta (1+\delta)}  \LSB  (1-\delta) + \frac{  \delta^2  }{\LB 1+ \delta\RB}\RSB}{ \LB 1-\frac{1}{\beta}   \frac{\delta^2}{\LB 1+ \delta\RB^2}   \RB }\label{eq:main_6}\\
\Gamma^o_{j,j'}&\triangleq  \frac {\frac{ \sqrt{\LB 1-(\sigma^{(j)})^2\RB\LB 1-(\sigma^{(j')})^2\RB}\delta^2}{\beta (1+\delta)}  \LSB  (1-\delta) + \frac{  \delta^2  }{\LB 1+ \delta\RB}\RSB}{ \LB 1-\frac{1}{\beta}  \LB 1-(\sigma^{(j)})^2\RB\LB 1-(\sigma^{(j')})^2\RB\frac{\delta^2}{\LB 1+ \delta\RB^2}   \RB }.
\label{eq:main_7}
\end{align} 
\end{theorem}

%%%%%%%%%%%%%%%%%%%%%%%%%%%%%%%%%%%%%%%%%%%%%%
%%%%%%%%%%%%%%%%%%%%%%%%%%%%%%%%%%%%%%%%%%%%%%
%%%%%%%%%%%%%%%%%%%%%%%%%%%%%%%%%%%%%%%%%%%%%%
%%%%%%%%%%%%%%%%%%%%%%%%%%%%%%%%%%%%%%%%%%%%%%
\section{Proof of Theorem~\ref{theorem}}\label{se:proof}
%%%%%%%%%%%%%%%%%%%%%%%%%%%%%%%%%%%%%%%%%%%%%%
%%%%%%%%%%%%%%%%%%%%%%%%%%%%%%%%%%%%%%%%%%%%%%
%%%%%%%%%%%%%%%%%%%%%%%%%%%%%%%%%%%%%%%%%%%%%%
%%%%%%%%%%%%%%%%%%%%%%%%%%%%%%%%%%%%%%%%%%%%%%
Our calculation is built upon results from both \cite{Wagner2012} and \cite{Muller2013}. We also make extensive use of classical RMT lemmas recalled in Appendix~\ref{app:literature}. Note that Lemma~\ref{lemma1} and Lemma~\ref{lemma2} are novel. In particular, Lemma~\ref{lemma1} extends \cite[Lemma~$15$]{Muller2013} and is an interesting result in itself. 

During the calculation we use the notation $x \asymp y$ to denote that $x-y\xrightarrow[K,M\rightarrow \infty]{a.s.}0$.
%%%%%%%%%%%%%%%%%%%%%%%%%%%%%%%%%%%%%%%%%%%%%%
%%%%%%%%%%%%%%%%%%%%%%%%%%%%%%%%%%%%%%%%%%%%%%
\subsection{Deterministic Equivalent for $\Psi^{(j)}$}\label{se:proof:preliminaries}
%%%%%%%%%%%%%%%%%%%%%%%%%%%%%%%%%%%%%%%%%%%%%%
%%%%%%%%%%%%%%%%%%%%%%%%%%%%%%%%%%%%%%%%%%%%%%

We start by finding a deterministic equivalent for $\Psi^{(j)}$. In fact, a deterministic equivalent for $\Psi^{(j)}$ is provided in \cite{Wagner2012}. However, it can also be obtained using Lemma~\ref{lemma1} with~$\sigma^{(j)}=\sigma^{(j')}=0$, which gives
\begin{equation}
\Psi^{(j)}\asymp\Gamma^o.
\label{eq:main_proof_1}
\end{equation}
Looking at the definition of $\Gamma^o$ in \eqref{eq:main_6}, it can be noted that, as expected, this deterministic equivalent does not depend on~$\sigma^{(j)}$.
 
%%%%%%%%%%%%%%%%%%%%%%%%%%%%%%%%%%%%%%%%%%%%%%
%%%%%%%%%%%%%%%%%%%%%%%%%%%%%%%%%%%%%%%%%%%%%%
\subsection{Deterministic Equivalent for $\bh_k^\He\bt^{\DCSI}_{\rZF,k}$:}\label{se:proof:preliminaries}
%%%%%%%%%%%%%%%%%%%%%%%%%%%%%%%%%%%%%%%%%%%%%%
%%%%%%%%%%%%%%%%%%%%%%%%%%%%%%%%%%%%%%%%%%%%%%
Turning to the desired signal at RX~$k$, we can write 
\begin{align}
\bh_k^\He\bt_{\rZF,k}^{\DCSI}
&=  \sum_{j=1}^n \frac{1}{M}\frac{\sqrt{P}}{\sqrt{\Psi^{(j)}}}\bh_k^\He\bE_j\bE_j^{\He}(\bC^{(j)})^{-1}\hat{\bh}^{(j)}_k\\
&\stackrel{(a)}{\asymp} \sqrt{\frac{P}{\Gamma^o}}\sum_{j=1}^n \frac{\frac{1}{M}\bh_k^\He\bE_j\bE_j^{\He}(\bC_{[k]}^{(j)})^{-1}\hat{\bh}^{(j)}_k}{1+\frac{1}{M}\bh_k^\He (\bC_{[k]}^{(j)})^{-1} \bh_k}\\ 
&\stackrel{(b)}{\asymp}  \sqrt{\frac{P}{\Gamma^o}}\sum_{j=1}^n \sqrt{1-(\sigma^{(j)})^2}  \frac{\frac{1}{M}\bh_k^\He\bE_j\bE_j^{\He}(\bC_{[k]}^{(j)})^{-1}\bh_k}{1+\frac{1}{M}\bh_k^\He (\bC_{[k]}^{(j)})^{-1} \bh_k}\\ 
& \stackrel{(c)}{\asymp}\sqrt{\frac{P}{\Gamma^o}} \sum_{j=1}^n  \sqrt{1-(\sigma^{(j)})^2} \frac{ \frac{1}{M}\trace \LB \bE_j \bE_j^{\He}(\bC_{[k]}^{(j)})^{-1}\RB}{1+\frac{1}{M} \trace \LB  (\bC_{[k]}^{(j)})^{-1} \RB} \\	
%& \stackrel{(d)}{\asymp}\sqrt{\frac{P}{\Gamma^o}} \sum_{j=1}^n  \sqrt{1-(\sigma^{(j)})^2} \frac{\frac{1}{M}\trace \LB \bE_j \bE_j^{\He}\bC^{-1}\RB}{1+\frac{1}{M} \trace \LB  \bC^{-1} \RB} \\
&\stackrel{(d)}{\asymp}\sqrt{\frac{P}{\Gamma^o}} \frac{1}{n}\sum_{j=1}^n  \sqrt{1-(\sigma^{(j)})^2} \frac{\delta}{1+\delta}
\label{eq:main_proof_2}
\end{align} 
where we have defined
\begin{align}
\bC_{[k]}^{(j)}\triangleq  \frac{ \hat{\bH}_{[k]}^{(j)}(\hat{\bH}_{[k]}^{(j)})^{\He}}{M}+\alpha\I_{M},\qquad \forall j
\label{eq:main_proof_3}
\end{align}
with 
\begin{align}
(\hat{\bH}_{[k]}^{(j)})^{\He}\triangleq   \begin{bmatrix}
\hat{\bh}_1^{(j)} &\hdots&\hat{\bh}_{k-1}^{(j)} &\hat{\bh}_{k+1}^{(j)} &\hdots&\hat{\bh}_K^{(j)} 
\end{bmatrix},\qquad \forall j.
\label{eq:main_proof_4}
\end{align}
Equality $(a)$ follows then from Lemma~\ref{lemma_resolvent} and the use of the deterministic equivalent derived for~$\Psi^{(j)}$, $(b)$ from Lemma~\ref{lemma_zero}, $(c)$ from Lemma~\ref{lemma_trace}, $(d)$ from Lemma~\ref{lemma_rank1} and the fundamental Theorem~\ref{fundamental_theorem}. Note that $\delta$ is defined in Theorem~\ref{fundamental_theorem}. It follows then directly that
\begin{align} 
\left |\bh_k^\He\bt_{\rZF,k}^{\DCSI}\right |^2&\asymp \frac{P}{\Gamma^o}\LB \frac{1}{n}\sum_{j=1}^n \sqrt{1-(\sigma^{(j)})^2}\RB^2 \frac{\delta^2}{\LB 1+\delta\RB^2}.
\label{eq:main_proof_5}
\end{align}   
 
%%%%%%%%%%%%%%%%%%%%%%%%%%%%%%%%%%%%%%%%%%%%%%
%%%%%%%%%%%%%%%%%%%%%%%%%%%%%%%%%%%%%%%%%%%%%%
\subsection{Deterministic Equivalent for the Interference Term}\label{se:proof:preliminaries}
%%%%%%%%%%%%%%%%%%%%%%%%%%%%%%%%%%%%%%%%%%%%%%
%%%%%%%%%%%%%%%%%%%%%%%%%%%%%%%%%%%%%%%%%%%%%% 
Our first step is to write explicitly the interference term using the definition of~$\bT^{\DCSI}$ in \eqref{eq:main_6} and replace~$\Psi^{(j)}$ by its deterministic equivalent.
\begin{align} 
\mathcal{I}_k&\triangleq \sum_{\ell=1,\ell\neq k}^K|\bh_k^{\He}\bt_{\rZF,\ell}^{\DCSI}|^2\\
&=\bh_k^{\He}\bT_{\rZF}^{\DCSI}  (\bT_{\rZF}^{\DCSI})^{\He}\bh_k-\bh_k^{\He}\bt_{\rZF,k}^{\DCSI}(\bt_{\rZF,k}^{\DCSI})^{\He}\bh_k\\ 
&=\frac{1}{M^2}\sum_{j=1}^n\sum_{j'=1}^n  \frac{P}{\sqrt{\Psi^{(j)}}\sqrt{\Psi^{(j')}}}\bh_k^{\He} \bE_j\bE_j^{\He}(\bC^{(j)})^{-1}(\hat{\bH}_{[k]}^{(j)})^{\He} \hat{\bH}_{[k]}^{(j')}(\bC^{(j')})^{-1}\bE_{j'}\bE_{j'}^{\He} \bh_k\\
&\asymp \frac{P}{\Gamma^o}\frac{1}{M^2} \sum_{j=1}^n\sum_{j'=1}^n \bh_k^{\He}\bE_j\bE_j^{\He}(\bC_{[k]}^{(j)})^{-1}(\hat{\bH}_{[k]}^{(j)})^{\He}\hat{\bH}_{[k]}^{(j')}(\bC^{({j'})})^{-1} \bE_{j'}\bE_{j'}^{\He}\bh_k\\ 
&+\frac{P}{\Gamma^o}\frac{1}{M^2} \sum_{j=1}^n\sum_{j'=1}^n \bh_k^{\He} \bE_j\bE_j^{\He}\LB(\bC^{(j)})^{-1}-(\bC_{[k]}^{(j)})^{-1}\RB(\hat{\bH}_{[k]}^{(j)})^{\He}\hat{\bH}_{[k]}^{(j')}(\bC^{({j'})})^{-1}\bE_{j'}\bE_{j'}^{\He} \bh_k.
\label{eq:main_proof_6}
\end{align}
To obtain a deterministic equivalent for the second summation in \eqref{eq:main_proof_6} we use the following relation 
\begin{align} 
(\bC^{(j)})^{-1}-(\bC_{[k]}^{(j)})^{-1}&=(\bC^{(j)})^{-1} \LB \bC_{[k]}^{(j)}-\bC^{(j)} \RB (\bC_{[k]}^{(j)})^{-1}\\
&=\!-\!(\bC^{(j)})^{-1} \!\LB\! c^{(j)}_0 \bh_k \bh_k^{\He}\!+\!c^{(j)}_1 \bm{\delta}_k^{(j)} (\bm{\delta}_k^{(j)})^{\He}\!+\!c^{(j)}_2 \bm{\delta}_k^{(j)} \bh_k^{\He}\!+\!c^{(j)}_2 \bh_k (\bm{\delta}_k^{(j)})^{\He} \RB (\bC_{[k]}^{(j)})^{-1}
\label{eq:main_proof_7} 
\end{align}
Inserting \eqref{eq:main_proof_7} in \eqref{eq:main_proof_6} yields
\begin{align} 
&\tilde{\mathcal{I}}_k\notag\\	
&\asymp \frac{P}{\Gamma^o}\frac{1}{M^2} \sum_{j=1}^n\sum_{j'=1}^n \bh_k^{\He}\bE_j\bE_j^{\He}(\bC_{[k]}^{(j)})^{-1}(\hat{\bH}_{[k]}^{(j)})^{\He}\hat{\bH}_{[k]}^{(j')}(\bC^{({j'})})^{-1} \bE_{j'}\bE_{j'}^{\He}\bh_k\notag\\ 
&-\frac{P}{\Gamma^o}\frac{1}{M^3} \sum_{j=1}^n\sum_{j'=1}^n \bh_k^{\He} \bE_j\bE_j^{\He} (\bC^{(j)})^{-1}  \LSB   \bh_k c^{(j)}_0 \bh_k^{\He}  \RSB  (\bC_{[k]}^{(j)})^{-1}    (\hat{\bH}_{[k]}^{(j)})^{\He}\hat{\bH}_{[k]}^{(j')}(\bC^{({j'})})^{-1}\bE_{j'}\bE_{j'}^{\He} \bh_k\notag\\
&-\frac{P}{\Gamma^o}\frac{1}{M^3} \sum_{j=1}^n\sum_{j'=1}^n \bh_k^{\He} \bE_j\bE_j^{\He}  (\bC^{(j)})^{-1} \LSB  \bm{\delta}_k^{(j)} c^{(j)}_1(\bm{\delta}_k^{(j)})^{\He}\RSB   (\bC_{[k]}^{(j)})^{-1}     (\hat{\bH}_{[k]}^{(j)})^{\He}\hat{\bH}_{[k]}^{(j')}(\bC^{({j'})})^{-1}\bE_{j'}\bE_{j'}^{\He} \bh_k\notag\\
&-\frac{P}{\Gamma^o}\frac{1}{M^3} \sum_{j=1}^n\sum_{j'=1}^n \bh_k^{\He} \bE_j\bE_j^{\He}  (\bC^{(j)})^{-1} \LSB   \bm{\delta}^{(j)}_k c^{(j)}_2 \bh_k^{\He}     \RSB(\bC_{[k]}^{(j)})^{-1} (\hat{\bH}_{[k]}^{(j)})^{\He}\hat{\bH}_{[k]}^{(j')}(\bC^{({j'})})^{-1}\bE_{j'}\bE_{j'}^{\He} \bh_k\notag\\
&-\frac{P}{\Gamma^o}\frac{1}{M^3} \sum_{j=1}^n\sum_{j'=1}^n \bh_k^{\He} \bE_j\bE_j^{\He}  (\bC^{(j)})^{-1}  \LSB  \bh_k c^{(j)}_2(\bm{\delta}_k^{(j)})^{\He}  \RSB   (\bC_{[k]}^{(j)})^{-1} (\hat{\bH}_{[k]}^{(j)})^{\He}\hat{\bH}_{[k]}^{(j')}(\bC^{({j'})})^{-1}\bE_{j'}\bE_{j'}^{\He} \bh_k\\
&\triangleq A-B-C-D-E.
\label{eq:main_proof_8} 
\end{align}  
We proceed by calculating each of the $5$~terms in \eqref{eq:main_proof_8} successively, using in particular Lemma~\ref{lemma2}:
\begin{align}  
A&=\frac{P}{\Gamma^o}\frac{1}{M^2} \sum_{j=1}^n\sum_{j'=1}^n \bh_k^{\He}\bE_j\bE_j^{\He}(\bC_{[k]}^{(j)})^{-1}(\hat{\bH}_{[k]}^{(j)})^{\He}\hat{\bH}_{[k]}^{(j')}(\bC^{({j'})})^{-1} \bE_{j'}\bE_{j'}^{\He}\bh_k\\
&\asymp \frac{P}{\Gamma^o}\sum_{j=1}^n\sum_{j'=1}^n  \frac{\trace \LB \bE_{j'}\bE_{j'}^{\He}   \bE_j\bE_j^{\He}(\bC^{(j)}_{[k]})^{-1}(\hat{\bH}_{[k]}^{(j)})^{\He}\bP_{[k]}\hat{\bH}_{[k]}^{(j')}(\bC_{[k]}^{({j'})})^{-1} \RB}{M^2}\notag\\
&-c_0^{(j')}\frac{\trace \LB \bE_j\bE_j^{\He}(\bC^{(j)}_{[k]})^{-1}(\hat{\bH}_{[k]}^{(j)})^{\He}\bP_{[k]}\hat{\bH}_{[k]}^{(j')}(\bC_{[k]}^{({j'})})^{-1}\RB}{M^2}   \frac{\trace \LB \bE_{j'}\bE_{j'}^{\He} (\bC_{[k]}^{({j'})})^{-1}\RB}{M}  \frac{1+c_1^{(j')} \frac{\trace \LB  (\bC_{[k]}^{({j'})})^{-1}\RB}{M}}{1+\frac{\trace \LB  (\bC_{[k]}^{({j'})})^{-1}\RB}{M}}\notag\\
&+ (c_2^{(j')})^2\frac{\trace \LB  \bE_j\bE_j^{\He}(\bC^{(j)}_{[k]})^{-1}(\hat{\bH}_{[k]}^{(j)})^{\He}\bP_{[k]}\hat{\bH}_{[k]}^{(j')}(\bC_{[k]}^{({j'})})^{-1}\RB}{M^2}  \frac{\trace \LB \bE_{j'}\bE_{j'}^{\He}(\bC_{[k]}^{({j'})})^{-1}\RB}{M}  \frac{\frac{\trace \LB  (\bC_{[k]}^{({j'})})^{-1}\RB}{M}}{1+\frac{\trace \LB (\bC_{[k]}^{({j'})})^{-1}\RB}{M}}.
\label{eq:main_proof_9} 
\end{align}  
From the unitary invariance of the distribution of~$\bH$ and $\bm{\Delta}^{(j)}$, it can be shown that

\begin{align}  
\frac{\trace \LB \bE_j\bE_j^{\He}(\bC^{(j)}_{[k]})^{-1}(\hat{\bH}_{[k]}^{(j)})^{\He}\bP_{[k]}\hat{\bH}_{[k]}^{(j')}(\bC_{[k]}^{({j'})})^{-1}\RB}{M^2}
&=\frac{1}{n}\frac{\trace \LB(\bC^{(j)}_{[k]})^{-1}(\hat{\bH}_{[k]}^{(j)})^{\He}\bP_{[k]}\hat{\bH}_{[k]}^{(j')}(\bC_{[k]}^{({j'})})^{-1}\RB}{M^2}\\
&\asymp \frac{1}{n}\Gamma^o_{j,j'}
\label{eq:main_proof_10} 
\end{align}  
where the last equality follows directly from applying Lemma~\ref{lemma1}. Inserting \eqref{eq:main_proof_10} in \eqref{eq:main_proof_9} and using the fundamental Theorem~\ref{fundamental_theorem} yields
\begin{align}  
A&\asymp \frac{P}{\Gamma^o}\sum_{j=1}^n\sum_{j'=1}^n  \frac{\Gamma^o}{n}\mathbb{1}_{j=j'}-c_0^{(j')}\frac{\Gamma^o_{j,j'}}{n}   \frac{\delta}{n}  \frac{1+c_1^{(j')} \delta}{1+\delta}+ (c_2^{(j')})^2\frac{\Gamma^o_{j,j'}}{n}   \frac{\delta}{n}   \frac{\delta}{ 1+\delta}.
\label{eq:main_proof_11} 
\end{align}  
We then proceed similarly for the remaining~$4$ terms:
\begin{align}  
B&=\frac{P}{\Gamma^o}\frac{1}{M^3} \sum_{j=1}^n\sum_{j'=1}^n \LSB \bh_k^{\He} \bE_j\bE_j^{\He}   (\bC^{(j)})^{-1}  \bh_k \RSB  \cdot c^{(j)}_0 \cdot\LSB \bh_k^{\He}     (\bC_{[k]}^{({j})})^{-1}  (\hat{\bH}_{[k]}^{(j)})^{\He}\hat{\bH}_{[k]}^{(j')}(\bC^{({j'})})^{-1}\bE_{j'}\bE_{j'}^{\He} \bh_k\RSB \\
&\asymp\frac{P}{\Gamma^o} \sum_{j=1}^n\sum_{j'=1}^n c^{(j)}_0  \LSB \frac{\trace \LB  \bE_j\bE_j^{\He}   (\bC_{[k]}^{(j)})^{-1} \RB}{M} \frac{1+c_1^{(j)} \frac{\trace \LB   (\bC_{[k]}^{(j)})^{-1}\RB}{M}}{1+\frac{\trace \LB (\bC_{[k]}^{(j)})^{-1}\RB}{M}}  \RSB\notag\\
&\cdot \bigg [ \frac{\trace \LB (\bC^{(j)}_{[k]})^{-1}(\hat{\bH}_{[k]}^{(j)})^{\He}\bP_{[k]}\hat{\bH}_{[k]}^{(j')}(\bC_{[k]}^{({j'})})^{-1} \bE_{j'}\bE_{j'}^{\He} \RB}{M^2}  \notag\\
&-c_0^{(j')} \frac{\trace\LB  (\bC_{[k]}^{(j)})^{-1}  (\hat{\bH}_{[k]}^{(j)})^{\He}\bP_{[k]}\hat{\bH}_{[k]}^{(j')}(\bC_{[k]}^{({j'})})^{-1} \RB }{M^2}     \frac{\trace\LB \bE_{j'}\bE_{j'}^{\He} (\bC_{[k]}^{({j'})})^{-1} \RB}{M} \frac{1+c_1^{(j')} \frac{\trace \LB (\bC_{[k]}^{(j)})^{-1}\RB}{M}}{1+\frac{\trace \LB  (\bC_{[k]}^{(j)})^{-1}\RB}{M}}   \notag\\
&+(c_2^{(j')})^2\frac{\trace\LB  (\bC_{[k]}^{(j)})^{-1}  (\hat{\bH}_{[k]}^{(j)})^{\He}\bP_{[k]}\hat{\bH}_{[k]}^{(j')}(\bC_{[k]}^{({j'})})^{-1} \RB }{M^2}     \frac{\trace\LB \bE_{j'}\bE_{j'}^{\He}  (\bC_{[k]}^{({j'})})^{-1} \RB}{M} \frac{\frac{\trace \LB  \bC^{-1}\RB}{M}}{1+\frac{\trace \LB  \bC^{-1}\RB}{M}}  \bigg]\\
&\asymp \frac{P}{\Gamma^o}\sum_{j=1}^n\sum_{j'=1}^n c^{(j)}_0    \frac{\delta}{n} \frac{1+c_1^{(j)} \delta}{1+\delta}   \LB \frac{\Gamma^o_{j,j'}}{n} -c_0^{(j')} \Gamma^o_{j,j'}     \frac{\delta}{n} \frac{1+c_1^{(j')} \delta}{1+\delta} +(c_2^{(j')})^2\Gamma^o_{j,j'}    \frac{\delta}{n} \frac{\delta}{1+\delta}  \RB.
\label{eq:main_proof_12} 
\end{align}  

\begin{align}  
C&=\frac{P}{\Gamma^o}\frac{1}{M^3} \sum_{j=1}^n\sum_{j'=1}^n \!\LSB\! \bh_k^{\He} \bE_j\bE_j^{\He}  (\bC^{(j)})^{-1}    \bm{\delta}_k^{(j)} \RSB\!\cdot \! c^{(j)}_1\!\cdot \!\LSB \!\bm{(\delta}_k^{(j)})^{\He}   (\bC_{[k]}^{(j)})^{-1}     (\hat{\bH}_{[k]}^{(j)})^{\He}\hat{\bH}_{[k]}^{(j')}(\bC^{({j'})})^{-1}\bE_{j'}\bE_{j'}^{\He} \bh_k\RSB\\
&=\frac{P}{\Gamma^o}\frac{1}{M^3} \sum_{j=1}^n \LSB \bh_k^{\He} \bE_j\bE_j^{\He}  (\bC^{(j)})^{-1}    \bm{\delta}_k^{(j)} \RSB\cdot  c^{(j)}_1\cdot \LSB \bm{(\delta}_k^{(j)})^{\He}   (\bC_{[k]}^{(j)})^{-1}     (\hat{\bH}_{[k]}^{(j)})^{\He}\bH_{[k]}^{({j})}(\bC^{({j})})^{-1}\bE_{j}\bE_{j}^{\He} \bh_k\RSB\\
&\asymp  \frac{P}{\Gamma^o} \sum_{j=1}^n c^{(j)}_1   \LSB - c^{(j)}_2 \frac{\trace \LB \bE_j\bE_j^{\He}   (\bC_{[k]}^{(j)})^{-1} \RB}{M} \frac{\frac{\trace \LB  (\bC_{[k]}^{(j)})^{-1}\RB}{M}}{1+\frac{\trace \LB  (\bC_{[k]}^{(j)})^{-1}\RB}{M}} \RSB\notag\\
&\cdot \bigg [ c_1^{(j)}c_2^{(j)} \frac{\trace\LB  (\bC_{[k]}^{(j)})^{-1}  (\hat{\bH}_{[k]}^{(j)})^{\He}\bP_{[k]}\hat{\bH}_{[k]}^{(j)}(\bC_{[k]}^{(j)})^{-1} \RB }{M^2}     \frac{\trace\LB \bE_j\bE_j^{\He} (\bC_{[k]}^{(j)})^{-1}  \RB}{M} \frac{\frac{\trace \LB  (\bC_{[k]}^{(j)})^{-1}\RB}{M}}{1+\frac{\trace \LB  (\bC_{[k]}^{(j)})^{-1}\RB}{M}} \notag  \\
&-c_2^{(j)} \frac{\trace\LB  (\bC_{[k]}^{(j)})^{-1}  (\hat{\bH}_{[k]}^{(j)})^{\He}\bP_{[k]}\hat{\bH}_{[k]}^{(j)}(\bC_{[k]}^{(j)})^{-1} \RB }{M^2}     \frac{\trace\LB  \bE_j\bE_j^{\He} (\bC_{[k]}^{(j)})^{-1} \RB}{M} \frac{1+c_1^{(j)}\frac{\trace \LB  (\bC_{[k]}^{(j)})^{-1}\RB}{M}}{1+\frac{\trace \LB   (\bC_{[k]}^{(j)})^{-1}\RB}{M}}  \bigg]\\
&\asymp  \frac{P}{\Gamma^o} \sum_{j=1}^n c^{(j)}_1  (-1) c^{(j)}_2 \frac{\delta}{n} \frac{\delta}{1+\delta} \LB c_1^{(j)}c_2^{(j)} \Gamma^o      \frac{\delta}{n} \frac{\delta}{1+\delta} -c_2^{(j)} \Gamma^o     \frac{\delta}{n} \frac{1+c_1^{(j)}\delta}{1+\delta}\RB.
\label{eq:main_proof_13} 
\end{align} 

\begin{align}  
D&=\frac{P}{\Gamma^o}\frac{1}{M^3} \sum_{j=1}^n\sum_{j'=1}^n \LSB \bh_k^{\He} \bE_j\bE_j^{\He}  (\bC^{(j)})^{-1} \bm{\delta}^{(j)}_k \RSB \cdot  c^{(j)}_2 \!\cdot \!\LSB \!\bh_k^{\He} (\bC_{[k]}^{(j)})^{-1} (\hat{\bH}_{[k]}^{(j)})^{\He}\hat{\bH}_{[k]}^{(j')}(\bC^{({j'})})^{-1}\bE_{j'}\bE_{j'}^{\He} \bh_k\RSB \notag\\ 
&\asymp\frac{P}{\Gamma^o} \sum_{j=1}^n \sum_{j'=1}^n  c^{(j)}_2\LSB  (-1) c^{(j)}_2 \frac{\trace \LB \bE_j\bE_j^{\He}   (\bC_{[k]}^{(j)})^{-1} \RB}{M} \frac{\frac{\trace \LB  (\bC_{[k]}^{(j)})^{-1}\RB}{M}}{1+\frac{\trace \LB  (\bC_{[k]}^{(j)})^{-1}\RB}{M}} \RSB\notag\\
&\cdot \bigg [ \frac{\trace \LB\bE_{j'}\bE_{j'}^{\He}  (\bC^{(j)}_{[k]})^{-1}(\hat{\bH}_{[k]}^{(j)})^{\He}\bP_{[k]}\hat{\bH}_{[k]}^{(j')}(\bC_{[k]}^{({j'})})^{-1} \RB}{M^2}  \notag\\
&-c_0^{(j')} \frac{\trace\LB   (\bC_{[k]}^{(j)})^{-1}  (\hat{\bH}_{[k]}^{(j)})^{\He}\bP_{[k]}\hat{\bH}_{[k]}^{(j')}(\bC_{[k]}^{({j'})})^{-1} \RB }{M^2}     \frac{\trace\LB \bE_{j'}\bE_{j'}^{\He} (\bC_{[k]}^{({j'})})^{-1}  \RB}{M} \frac{1+c_1^{(j')} \frac{\trace \LB   \bC^{-1}\RB}{M}}{1+\frac{\trace \LB  \bC^{-1}\RB}{M}}   \notag\\
&+(c_2^{(j')})^2\frac{\trace\LB   (\bC_{[k]}^{(j)})^{-1}  (\hat{\bH}_{[k]}^{(j)})^{\He}\bP_{[k]}\hat{\bH}_{[k]}^{(j')}(\bC_{[k]}^{({j'})})^{-1} \RB }{M^2}     \frac{\trace\LB \bE_{j'}\bE_{j'}^{\He}   (\bC_{[k]}^{({j'})})^{-1} \RB}{M} \frac{\frac{\trace \LB   (\bC_{[k]}^{({j'})})^{-1}\RB}{M}}{1\!+\!\frac{\trace \LB  (\bC_{[k]}^{({j'})})^{-1}\RB}{M}}  \bigg ]\\
&\asymp\frac{P}{\Gamma^o} \sum_{j=1}^n \sum_{j'=1}^n  c^{(j)}_2 (-1) c^{(j)}_2 \frac{\delta}{n} \frac{\delta}{1+\delta} \LB \frac{\Gamma^o_{j,j'}}{n} -c_0^{(j')}  \Gamma^o_{j,j'}      \frac{\delta}{n} \frac{1+c_1^{(j')} \delta}{1+\delta} +(c_2^{(j')})^2\Gamma^o_{j,j'}     \frac{\delta}{n} \frac{\delta}{1\!+\!\delta}  \RB.
\label{eq:main_proof_14} 
\end{align}

\begin{align}  
E&=\!\frac{P}{\Gamma^o}\frac{1}{M^3} \sum_{j=1}^n\sum_{j'=1}^n \LSB \bh_k^{\He} \bE_j\bE_j^{\He}  (\bC^{(j)})^{-1}     \bh_k\!\RSB \!\cdot\!  c^{(j)}_2\cdot \LSB (\bm{\delta}_k^{(j)})^{\He}    (\bC_{[k]}^{(j)})^{-1} (\hat{\bH}_{[k]}^{(j)})^{\He}\hat{\bH}_{[k]}^{(j')}(\bC^{({j'})})^{-1}\bE_{j'}\bE_{j'}^{\He} \bh_k\!\RSB \\
&=\frac{P}{\Gamma^o}\frac{1}{M^3}\! \sum_{j=1}^n \!\LSB \bh_k^{\He} \bE_j\bE_j^{\He}  (\bC^{(j)})^{-1}     \bh_k\!\RSB \!\cdot \! c^{(j)}_2\cdot \!\LSB\! (\bm{\delta}_k^{(j)})^{\He}    (\bC_{[k]}^{(j)})^{-1} (\hat{\bH}_{[k]}^{(j)})^{\He}\bH_{[k]}^{({j})}(\bC^{({j})})^{-1}\bE_{j}\bE_{j}^{\He} \bh_k\RSB \\
&\asymp  \sum_{j=1}^n c^{(j)}_2  \LSB \frac{\trace \LB   \bE_j\bE_j^{\He}   (\bC_{[k]}^{(j)})^{-1} \RB}{M}  \frac{1+c^{(j)}_1\frac{\trace \LB (\bC_{[k]}^{(j)})^{-1}\RB}{M}}{1+\frac{\trace \LB  (\bC_{[k]}^{(j)})^{-1}\RB}{M}}  \RSB\notag\\
&\cdot \bigg [ c_1^{(j)}c_2^{(j)} \frac{\trace\LB (\bC_{[k]}^{(j)})^{-1}  (\hat{\bH}_{[k]}^{(j)})^{\He}\bP_{[k]}\hat{\bH}_{[k]}^{(j)}(\bC_{[k]}^{(j)})^{-1} \RB }{M^2}     \frac{\trace\LB\bE_j\bE_j^{\He}  (\bC_{[k]}^{(j)})^{-1}  \RB}{M} \frac{\frac{\trace \LB  (\bC_{[k]}^{(j)})^{-1}\RB}{M}}{1+\frac{\trace \LB  (\bC_{[k]}^{(j)})^{-1}\RB}{M}}  \notag \\
&-c_2^{(j)} \frac{\trace\LB   (\bC_{[k]}^{(j)})^{-1}  (\hat{\bH}_{[k]}^{(j)})^{\He}\bP_{[k]}\hat{\bH}_{[k]}^{(j)}(\bC_{[k]}^{(j)})^{-1} \RB }{M^2}     \frac{\trace\LB \bE_j\bE_j^{\He}  (\bC_{[k]}^{(j)})^{-1}  \RB}{M} \frac{1+c_1^{(j)} \frac{\trace \LB  (\bC_{[k]}^{(j)})^{-1}\RB}{M}}{1+\frac{\trace \LB  (\bC_{[k]}^{(j)})^{-1}\RB}{M}}  \bigg]\\
&\asymp  \sum_{j=1}^n c^{(j)}_2   \frac{\delta}{n}  \frac{1+c^{(j)}_1\delta}{1+\delta} \LB c_1^{(j)}c_2^{(j)}  \Gamma^o      \frac{\delta}{n} \frac{\delta}{1+\delta}  -c_2^{(j)} \Gamma^o     \frac{\delta}{n} \frac{1+c_1^{(j)} \delta}{1+\delta}  \RB .
\label{eq:main_proof_15} 
\end{align} 
The final expression is obtained after inserting all the deterministic equivalents derived inside the interference expression \eqref{eq:main_proof_8}. The compact expression of the theorem is obtained after algebraic manipulations using the software Mathematica.

\begin{remark}
It is important to differentiate the cases $j=j'$ and $j\neq j'$ when computing $\Gamma^o_{j,j'}$. Indeed, in the case $j=j'$, it holds
\begin{equation}
\Gamma^o_{j,j'}=\Gamma^o.
\label{eq:main_proof_16}
\end{equation}\qed
\end{remark}
\FloatBarrier
%%%%%%%%%%%%%%%%%%%%%%%%%%%%%%%%%%%%%%%%%%%%%%%%%%%%%%%%%%%%%%%%%%%%%%%%%%%%%%%%%
%%%%%%%%%%%%%%%%%%%%%%%%%%%%%%%%%%%%%%%%%%%%%%%%%%%%%%%%%%%%%%%%%%%%%%%%%%%%%%%%%
%%%%%%%%%%%%%%%%%%%%%%%%%%%%%%%%%%%%%%%%%%%%%%%%%%%%%%%%%%%%%%%%%%%%%%%%%%%%%%%%%
%%%%%%%%%%%%%%%%%%%%%%%%%%%%%%%%%%%%%%%%%%%%%%%%%%%%%%%%%%%%%%%%%%%%%%%%%%%%%%%%%
\section{Simulation Results}\label{se:simulations}
%%%%%%%%%%%%%%%%%%%%%%%%%%%%%%%%%%%%%%%%%%%%%%%%%%%%%%%%%%%%%%%%%%%%%%%%%%%%%%%%%
%%%%%%%%%%%%%%%%%%%%%%%%%%%%%%%%%%%%%%%%%%%%%%%%%%%%%%%%%%%%%%%%%%%%%%%%%%%%%%%%%
%%%%%%%%%%%%%%%%%%%%%%%%%%%%%%%%%%%%%%%%%%%%%%%%%%%%%%%%%%%%%%%%%%%%%%%%%%%%%%%%%
%%%%%%%%%%%%%%%%%%%%%%%%%%%%%%%%%%%%%%%%%%%%%%%%%%%%%%%%%%%%%%%%%%%%%%%%%%%%%%%%%

We now verify using Monte-Carlo simulations the accuracy of the asymptotic expression derived in Theorem~\ref{theorem}. We consider a network consisting of $n=3$~TXs with a sum power constraint given by $P=10$~dB and~$\alpha=1/P$. We focus in this work on the case of $(\sigma^{(j)})^2=0.1,\forall j=1,\ldots,n$ so as to emphasize the \emph{price of distributedness}.
\begin{figure}[htp!] 
\centering
\includegraphics[width=1\columnwidth]{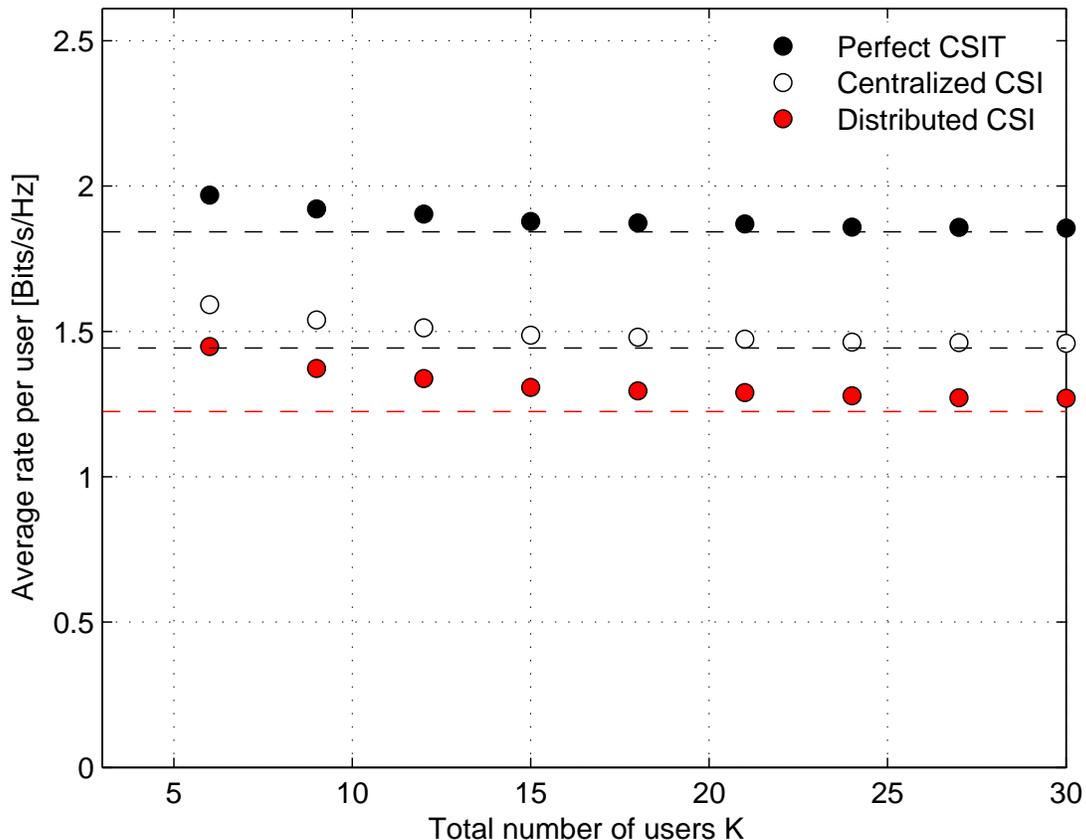}
\caption{Average rate per user as a function of the number of users~$K$ with $(\sigma^{(j)})^2=0.1,\forall j$.} 
\label{Rate_equal_01}
\end{figure}  

In Fig.~\ref{Rate_equal_01}, we show the rate per user as a function of the number of users for a square setting where $M=nM_{\TX}=K$ (i.e., $\beta=1$) in the distributed CSIT configuration where $(\sigma^{(j)})^2=0.1,\forall j=1,\ldots,n$. For comparison purpose, we also show the rate per user obtained in the case of imperfect centralized CSIT with $(\sigma^{\CCSI})^2=0.1$ and with perfect CSIT (i.e. $(\sigma^{\CCSI})^2=0$ or equivalently $(\sigma^{(j)})^2=0,\forall j$). As note earlier, a deterministic equivalent for the centralized case is obtained using $n=1$ in Theorem~\ref{theorem}. 

The large system deterministic equivalents are shown to be useful with just $20$ to $30$ users and antennas. In addition, the cost of having distributed information is also highlighted by the losses compared to the centralized configuration. This shows the necessity to take properly into account the CSI configuration when designing the feedback scheme and the precoder. 
\begin{figure}[htp!] 
\centering
\includegraphics[width=1\columnwidth]{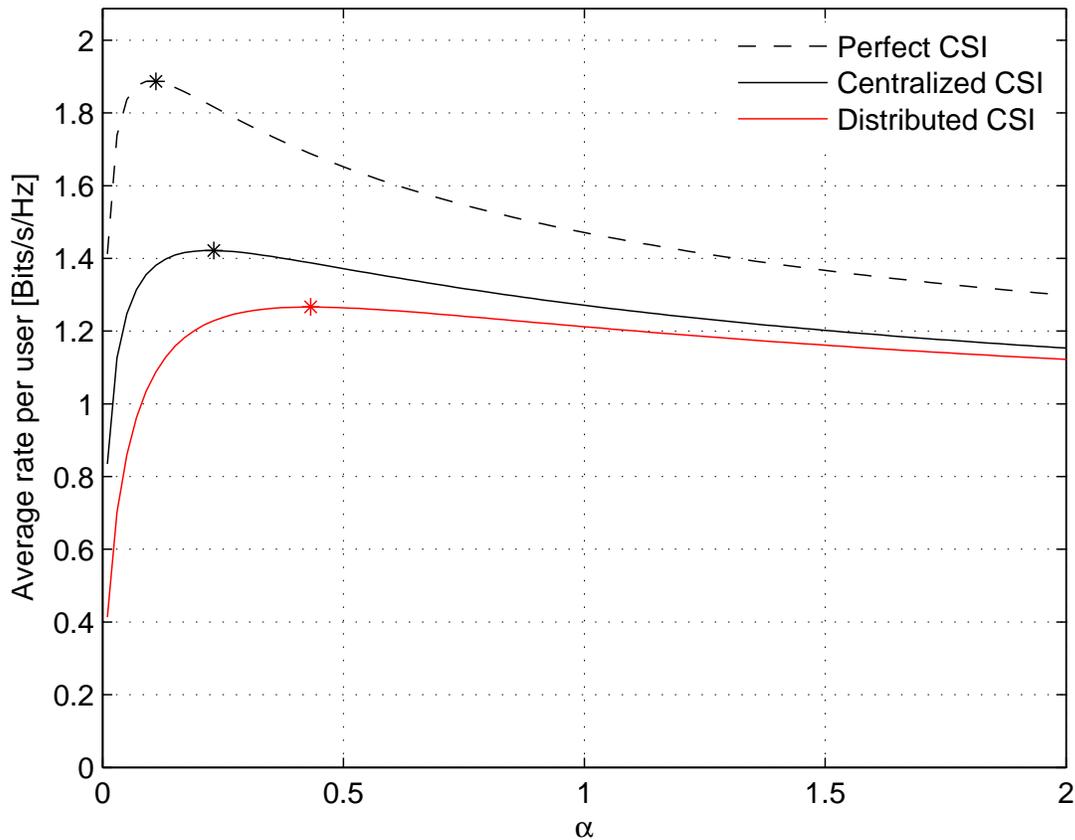}
\caption{Average rate per user as a function of~$\alpha$ for $(\sigma^{(j)})^2=0.1,\forall j$.} 
\label{Rate_optimization_alpha}
\end{figure}   
Considering the same network configuration, we show in Fig.~\ref{Rate_optimization_alpha} the average rate per user in terms of the regularization factor~$\alpha$. Quite interestingly it appears that the optimal regularization factors (represented with the * marker) are not the same in the centralized and distributed CSI settings. %The interest of the asymptotic expression is exactly to made very easy such parameter optimization instead of requiring expensive Monte-Carlo simulations.

%%%%%%%%%%%%%%%%%%%%%%%%%%%%%%%%%%%%%%%%%%%%%%%%%%%%%%%%%%%%%%%%%%%%%%%%%%%%%%%%%
%%%%%%%%%%%%%%%%%%%%%%%%%%%%%%%%%%%%%%%%%%%%%%%%%%%%%%%%%%%%%%%%%%%%%%%%%%%%%%%%%
%%%%%%%%%%%%%%%%%%%%%%%%%%%%%%%%%%%%%%%%%%%%%%%%%%%%%%%%%%%%%%%%%%%%%%%%%%%%%%%%%
%%%%%%%%%%%%%%%%%%%%%%%%%%%%%%%%%%%%%%%%%%%%%%%%%%%%%%%%%%%%%%%%%%%%%%%%%%%%%%%%%
\section{Conclusion}
%%%%%%%%%%%%%%%%%%%%%%%%%%%%%%%%%%%%%%%%%%%%%%%%%%%%%%%%%%%%%%%%%%%%%%%%%%%%%%%%%
%%%%%%%%%%%%%%%%%%%%%%%%%%%%%%%%%%%%%%%%%%%%%%%%%%%%%%%%%%%%%%%%%%%%%%%%%%%%%%%%%
%%%%%%%%%%%%%%%%%%%%%%%%%%%%%%%%%%%%%%%%%%%%%%%%%%%%%%%%%%%%%%%%%%%%%%%%%%%%%%%%%
%%%%%%%%%%%%%%%%%%%%%%%%%%%%%%%%%%%%%%%%%%%%%%%%%%%%%%%%%%%%%%%%%%%%%%%%%%%%%%%%%
We have studied in this work the joint transmission using regularized ZF in a distributed CSI configuration. Using RMT tools, an analytical expression has been derived to approximate the average rate per user. This expression becomes asymptotically exact in the large system limit where the number of transmit antennas and the number of receive antennas go to infinity at the same pace. This new deterministic equivalent reveals the cost related to not just CSI feedback limitation, but also backhaul sharing limitations and can be helpful to design more robust systems. The extensions to more general channel and CSI models are challenging and subject to ongoing work. Note that the price of distributedness is evaluated here for a conventional precoder. This further motivates the development of novel precoding schemes that are more suitable to the distributed CSI setting.
%%%%%%%%%%%%%%%%%%%%%%%%%%%%%%%%%%%%%%%%%%%%%%%%%%%%%%%%%%%%%%%%%%%%%%%%%%%%%%%%%
%%%%%%%%%%%%%%%%%%%%%%%%%%%%%%%%%%%%%%%%%%%%%%%%%%%%%%%%%%%%%%%%%%%%%%%%%%%%%%%%%
%%%%%%%%%%%%%%%%%%%%%%%%%%%%%%%%%%%%%%%%%%%%%%%%%%%%%%%%%%%%%%%%%%%%%%%%%%%%%%%%%
%%%%%%%%%%%%%%%%%%%%%%%%%%%%%%%%%%%%%%%%%%%%%%%%%%%%%%%%%%%%%%%%%%%%%%%%%%%%%%%%%
\appendix

%%%%%%%%%%%%%%%%%%%%%%%%%%%%%%%%%%%%%%%%%%%%%%%%%%%%%%%%%%%%%%%%%%%%%%%%%%%%%%%%%
%%%%%%%%%%%%%%%%%%%%%%%%%%%%%%%%%%%%%%%%%%%%%%%%%%%%%%%%%%%%%%%%%%%%%%%%%%%%%%%%%
%%%%%%%%%%%%%%%%%%%%%%%%%%%%%%%%%%%%%%%%%%%%%%%%%%%%%%%%%%%%%%%%%%%%%%%%%%%%%%%%%
%%%%%%%%%%%%%%%%%%%%%%%%%%%%%%%%%%%%%%%%%%%%%%%%%%%%%%%%%%%%%%%%%%%%%%%%%%%%%%%%%
\subsection{Classical Lemmas from the Literature} \label{app:literature}
%%%%%%%%%%%%%%%%%%%%%%%%%%%%%%%%%%%%%%%%%%%%%%%%%%%%%%%%%%%%%%%%%%%%%%%%%%%%%%%%%
%%%%%%%%%%%%%%%%%%%%%%%%%%%%%%%%%%%%%%%%%%%%%%%%%%%%%%%%%%%%%%%%%%%%%%%%%%%%%%%%%
%%%%%%%%%%%%%%%%%%%%%%%%%%%%%%%%%%%%%%%%%%%%%%%%%%%%%%%%%%%%%%%%%%%%%%%%%%%%%%%%%
%%%%%%%%%%%%%%%%%%%%%%%%%%%%%%%%%%%%%%%%%%%%%%%%%%%%%%%%%%%%%%%%%%%%%%%%%%%%%%%%%

\begin{lemma}[Adapted from \cite{Wagner2012,Couillet2011}]
Let $\alpha>0$ and $(\delta_k)_{k \geq 0}$ be the sequence defined as 
\begin{equation}
\begin{cases}
\delta_0=\frac{1}{\alpha}&  \\
\delta_k=\frac{1}{M}\trace \LB\alpha \I_M+\frac{1}{\beta\LB 1+\delta_{k-1}\RB}\bI_M \RB^{-1}& \text{for $k\geq 1$}
\label{eq:fixed_point}
\end{cases}.
\end{equation}
Then $\delta_k\xrightarrow{k\rightarrow\infty}\delta$, with $\delta$ being by construction a fixed point of \eqref{eq:fixed_point}.
\end{lemma}

\begin{lemma}[Resolvent Identities \cite{Muller2013,Couillet2011}]
\label{lemma_resolvent}
Given any matrix~$\bH\in\mathbb{C}^{K\times M}$, let $\bh^{\He}_k$ denote its $k$th row and $\bH_k\in\mathbb{C}^{(K-1)\times M}$ denote the matrix obtained after removing the $k$th row from $\bH$. The resolvent matrices of $\bH$ and $\bH_k$ are denoted by $\bQ\triangleq \LB \bH^\He\bH+\alpha\I_M\RB^{-1}$ and $\bQ_k\triangleq \LB \bH_k^\He\bH_k+\alpha\I_M\RB^{-1}$, with $\alpha>0$, respectively. It then holds that
\begin{align}
\bQ&=\bQ_k-\frac{1}{M}\frac{\bQ_k\bh_k\bh_k^{\He}\bQ_k}{1+\frac{1}{M}\bh_k^{\He}\bQ_k\bh_k}
\end{align} 
and
\begin{align} 
\bh_k^{\He}\bQ&=\frac{\bh_k^{\He}\bQ_k}{1+\frac{1}{M}\bh_k^{\He}\bQ_k\bh_k} .
\end{align} 
\end{lemma}

\begin{lemma} [\cite{Muller2013,Couillet2011}]
\label{lemma_trace}
Let $(\bA_N)_{N\geq 1}, \bA_N\in \mathbb{C}^{N\times N}$ be a sequence of matrices such that $\lim\sup\|\bA_N\|<\infty$, and $(\xv_N)_{N\geq 1}, \xv_N\in \mathbb{C}^{N\times 1}$ be a sequence of random vectors of i.i.d. entries of zero mean, unit variance, and finite eighth order moment independent of~$\bA_N$. Then,
\begin{equation}
\frac{1}{N}\xv_N^{\He}\bA_N\xv_N-\frac{1}{N}\trace\LB \bA_N\RB \xrightarrow[N\rightarrow\infty]{a.s.}0.
\end{equation}
\end{lemma}

\begin{lemma}[\cite{Muller2013,Couillet2011}]
\label{lemma_zero}
Let $(\bA_N)_{N\geq 1}, \bA_N\in \mathbb{C}^{N\times N}$ be a sequence of matrices such that $\lim\sup\|\bA_N\|<\infty$, and $\xv_N,\yv_N$ be random, mutually independent with i.i.d. entries of zero mean, unit variance, finite eighth order moment, and independent of~$\bA_N$. Then, 
\begin{equation}
\frac{1}{N}\xv_N^{\He}\bA_N\yv_N\xrightarrow[N\rightarrow\infty]{a.s.}0.
\end{equation}
\end{lemma} 

\begin{lemma}  [\cite{Wagner2012,Couillet2011}]
\label{lemma_rank1}
Let $\bQ$ and $\bQ_k$ be as given in Lemma~\ref{lemma_resolvent}. Then, for any matrix $\bA$, we have
\begin{equation}
\trace\LB \bA\LB \bQ-\bQ_k\RB\RB\leq \|\bA\|_2.
\end{equation}
\end{lemma}

\begin{lemma} [\cite{Wagner2012,Couillet2011}]
\label{lemma_c_0}
Let $\bU,\bV,\bm{\Theta}$ be of uniformly bounded spectral norm with respect to $N$ and let $\bV$ be invertible. Further, define $\xv\triangleq \bm{\Theta}\bm{\bz}$ and $\yv\triangleq \bm{\Theta}\bm{\bq}$ where $\bz,\bq\in \mathbb{C}^{N}$ have i.i.d. complex entries of zero mean, variance $1/N$	and finite $8$th order moment and be mutually independent as well as independent of~$\bU,\bV$. Define $c_0,c_1,c_2\in \mathbb{R}^{+}$ such that $c_0c_1-c_2^2>0$, and let $u\triangleq \frac{1}{N}\trace \LB\bm{\Theta}\bV^{-1}\RB$ and $u'\triangleq \frac{1}{N}\trace \LB\bm{\Theta}\bU\bV^{-1}\RB$. Then we have:
\begin{align}
&\xv^{\He}\bU \LB \bV+c_0\xv \xv^{\He}+c_1\yv \yv^{\He}+c_2\xv \yv^{\He}+c_2\yv \xv^{\He} \RB^{-1} \xv-\frac{u'\LB 1+c_1 u \RB}{ (c_0c_1-c_2^2)u^2+(c_0+c_1)u+1}\rightarrow 0
\end{align} 
as well as
\begin{align}
&\xv^{\He}\bU \LB \bV+c_0\xv \xv^{\He}+c_1\yv \yv^{\He}+c_2\xv \yv^{\He}+c_2\yv \xv^{\He} \RB^{-1} \yv-\frac{-c_2 uu'}{ (c_0c_1-c_2^2)u^2+(c_0+c_1)u+1}\rightarrow 0
\end{align} 
\end{lemma}

%%%%%%%%%%%%%%%%%%%%%%%%%%%%%%%%%%%%%%%%%%%%%%%%%%%%%%%%%%%%%%%%%%%%%%%%%%%%%%%%%%%%%%%%%%%%%%%%%%%%%%%%%%%%%%%%%%%%%%
%%%%%%%%%%%%%%%%%%%%%%%%%%%%%%%%%%%%%%%%%%%%%%%%%%%%%%%%%%%%%%%%%%%%%%%%%%%%%%%%%%%%%%%%%%%%%%%%%%%%%%%%%%%%%%%%%%%%%%
%%%%%%%%%%%%%%%%%%%%%%%%%%%%%%%%%%%%%%%%%%%%%%%%%%%%%%%%%%%%%%%%%%%%%%%%%%%%%%%%%%%%%%%%%%%%%%%%%%%%%%%%%%%%%%%%%%%%%%
%%%%%%%%%%%%%%%%%%%%%%%%%%%%%%%%%%%%%%%%%%%%%%%%%%%%%%%%%%%%%%%%%%%%%%%%%%%%%%%%%%%%%%%%%%%%%%%%%%%%%%%%%%%%%%%%%%%%%% 
\subsection{New Lemmas} 
%%%%%%%%%%%%%%%%%%%%%%%%%%%%%%%%%%%%%%%%%%%%%%%%%%%%%%%%%%%%%%%%%%%%%%%%%%%%%%%%%%%%%%%%%%%%%%%%%%%%%%%%%%%%%%%%%%%%%%
%%%%%%%%%%%%%%%%%%%%%%%%%%%%%%%%%%%%%%%%%%%%%%%%%%%%%%%%%%%%%%%%%%%%%%%%%%%%%%%%%%%%%%%%%%%%%%%%%%%%%%%%%%%%%%%%%%%%%%
%%%%%%%%%%%%%%%%%%%%%%%%%%%%%%%%%%%%%%%%%%%%%%%%%%%%%%%%%%%%%%%%%%%%%%%%%%%%%%%%%%%%%%%%%%%%%%%%%%%%%%%%%%%%%%%%%%%%%%
%%%%%%%%%%%%%%%%%%%%%%%%%%%%%%%%%%%%%%%%%%%%%%%%%%%%%%%%%%%%%%%%%%%%%%%%%%%%%%%%%%%%%%%%%%%%%%%%%%%%%%%%%%%%%%%%%%%%%%

%%%%%%%%%%%%%%%%%%%%%%%%%%%%%%%%%%%%%%%%%%%%%%%%%%%%%%%%%%%%%%%%%%%%%%%%%%%%%%%%%%%%%%%%%%%%%%%%%%%%%%%%%%%%%%%%%%%%%%
%%%%%%%%%%%%%%%%%%%%%%%%%%%%%%%%%%%%%%%%%%%%%%%%%%%%%%%%%%%%%%%%%%%%%%%%%%%%%%%%%%%%%%%%%%%%%%%%%%%%%%%%%%%%%%%%%%%%%%
%%%%%%%%%%%%%%%%%%%%%%%%%%%%%%%%%%%%%%%%%%%%%%%%%%%%%%%%%%%%%%%%%%%%%%%%%%%%%%%%%%%%%%%%%%%%%%%%%%%%%%%%%%%%%%%%%%%%%%
%%%%%%%%%%%%%%%%%%%%%%%%%%%%%%%%%%%%%%%%%%%%%%%%%%%%%%%%%%%%%%%%%%%%%%%%%%%%%%%%%%%%%%%%%%%%%%%%%%%%%%%%%%%%%%%%%%%%%%
%%%%%%%%%%%%%%%%%%%%%%%%% FIRST LEMMA     %%%%%%%%%%%%%%%%%%%%%%%%%%%%%%%%%%%%%%%%%%%%%%%%%%%%%%%%%%%%%%%%%%%%%%%%%%%%
%%%%%%%%%%%%%%%%%%%%%%%%%%%%%%%%%%%%%%%%%%%%%%%%%%%%%%%%%%%%%%%%%%%%%%%%%%%%%%%%%%%%%%%%%%%%%%%%%%%%%%%%%%%%%%%%%%%%%%
%%%%%%%%%%%%%%%%%%%%%%%%%%%%%%%%%%%%%%%%%%%%%%%%%%%%%%%%%%%%%%%%%%%%%%%%%%%%%%%%%%%%%%%%%%%%%%%%%%%%%%%%%%%%%%%%%%%%%%
%%%%%%%%%%%%%%%%%%%%%%%%%%%%%%%%%%%%%%%%%%%%%%%%%%%%%%%%%%%%%%%%%%%%%%%%%%%%%%%%%%%%%%%%%%%%%%%%%%%%%%%%%%%%%%%%%%%%%%
%%%%%%%%%%%%%%%%%%%%%%%%%%%%%%%%%%%%%%%%%%%%%%%%%%%%%%%%%%%%%%%%%%%%%%%%%%%%%%%%%%%%%%%%%%%%%%%%%%%%%%%%%%%%%%%%%%%%%%

\begin{lemma}
\label{lemma1}
Let~$\bH'$ and $\bH''$ be two imperfect multi-user channel estimates as described in Section~\ref{se:SM}. Let $\bQ'\triangleq\LB\frac{\bH'^{\He}\bH'}{M} +\alpha\I_M\RB^{-1}$ and  $\bQ''\triangleq\LB\frac{\bH''^{\He}\bH''}{M} +\alpha\I_M\RB^{-1}$ with $\alpha>0$. Let $\bA \in \mathbb{C}^{M\times M}$ be of uniformly bounded spectral norm with respect to $M$. Then,
\begin{equation}
\begin{aligned}
&\frac{1}{M^2} \trace \LB\bA \bQ' \bH'^{\He}\bH'' \bQ''\RB-\frac{ \frac{1}{M}\trace \LB \bA \RB\delta^2\sqrt{c_0'c_0''}}{\beta (1+\delta)}  \LSB  (1-\delta) + \frac{\LB \delta^2+\sqrt{c_0'c_0''} Y_0\RB}{\LB 1+ \delta\RB}\RSB\xrightarrow{a.s.}0
\end{aligned}
\end{equation}
with $c_0'\triangleq 1-\sigma'^2$,  $c_0''\triangleq  1-\sigma''^2$, $\bQ_o$ defined as in Theorem~\ref{fundamental_theorem}, and $Y_0$ defined as
\begin{equation}
Y_0\triangleq\frac {\frac{ \sqrt{c_0'c_0''}\delta^2}{\beta (1+\delta)}  \LSB  (1-\delta) + \frac{  \delta^2  }{\LB 1+ \delta\RB}\RSB}{ \LB 1-\frac{1}{\beta}  c_0'c_0''\frac{\delta^2}{\LB 1+ \delta\RB^2}   \RB }.
\end{equation}
Note that in the case where $\bA=\bI_M$, the result simplifies to
\begin{equation}
\frac{1}{M^2} \trace \LB \bQ' \bH'^{\He}\bH'' \bQ''\RB-Y_0\xrightarrow{a.s.}0 .
\end{equation}
%%%%%%%Also, setting $\sigma^{(j)}=\sigma^{(j)}=0$ yields
%%%%%%%\begin{equation}
%%%%%%%\frac{1}{M^2} \trace \LB \bQ \bH^{\He}\bH \bQ\RB-{\frac{  \frac{1}{M}\trace \LB \bT^2 \RB}{\beta (1+\delta)}  \LSB  (1-\delta) + \frac{  \delta^2  }{\LB 1+ \delta\RB}\RSB}{ \LB 1-\frac{1}{\beta} }\frac{\frac{1}{M} \trace \LB \bT^2\RB}{\LB 1+ \delta\RB^2}   \RB }\xrightarrow{a.s.}0 .
%%%%%%%\end{equation}
\end{lemma}  
\begin{proof}
We start by defining
\begin{align}
\bQ'_{\ell}\triangleq  \LB \frac{\bH_{\ell}' \bH_{\ell}'^{\He}}{M}+\alpha\I_{M} \RB^{-1}
\label{eq:proof_lemma1_1}
\end{align}
with 
\begin{align}
\bH_{\ell}'^{\He}\triangleq   \begin{bmatrix}
\bh'_1 &\hdots&\bh'_{\ell-1} &\bh'_{\ell+1} &\hdots&\bh'_K 
\end{bmatrix}.
\label{eq:proof_lemma1_2}
\end{align}
We then define similarly $\bQ''_{\ell}$ and $\bH_{\ell}''^{\He}$. Let us start by writing the simple equality
\begin{align}
\bQ'-\bQ_o&=\bQ_o \LB \bQ_o^{-1}-\bQ'^{-1}\RB \bQ'\\
&=\bQ_o \LB \frac{\I_K}{\beta \LB 1+\delta\RB } -\frac{\bH'^{\He}\bH'}{M} \RB \bQ
\label{eq:proof_lemma1_3}
\end{align}
We can then replace $\bQ'$ using \eqref{eq:proof_lemma1_3} to obtain
\begin{align}
&\frac{1}{M^2} \trace \LB\bA \bQ' \bH'^{\He}\bH''\bQ''\RB\\
&= \frac{1}{M^2} \trace \LB\bA \bQ_o \bH'^{\He}\bH''\bQ''\RB \!+\!\frac{\trace \LB\bA \bQ_o\bQ' \bH'^{\He}\bH''\bQ''\RB}{M^2\beta \LB 1+\delta\RB }  \!-\!\frac{1}{M^3} \trace \LB\bA \bQ_o \bH'^{\He}\bH' \bQ'\bH'^{\He}\bH''\bQ''\RB\\
&\triangleq Z_1 +Z_2+Z_3.
\label{eq:proof_lemma1_4}
\end{align}
We will now calculate separately each of the term~$Z_i$. Starting with~$Z_1$ gives 
\begin{align}
Z_1&= \frac{1}{M^2} \trace \LB\bA \bQ_o \bH'^{\He}\bH''\bQ''\RB\\
&= \frac{1}{M} \sum_{\ell=1}^K\frac{1}{M}\bh_{\ell}''^{\He} \bQ''\bA \bQ_o \bh'_{\ell}\\
&\stackrel{(a)}{=} \frac{1}{M} \sum_{\ell=1}^K \frac{1}{M}\frac{\bh_{\ell}''^{\He}\bQ_{\ell}''\bA \bQ_o \bh_{\ell}'}{1+\frac{1}{M}\bh_{\ell}''^{\He} \bQ_{\ell}''\bh_{\ell}''}\\ 
&\stackrel{(b)}{\asymp} \frac{1}{M} \sum_{\ell=1}^K \frac{\sqrt{c_0'c_0''} \frac{1}{M}\trace \LB \bQ_{\ell}''\bA \bQ_o \RB}{1+\frac{1}{M}\trace \LB \bQ_{\ell}'' \RB}\\
&\stackrel{(c)}{\asymp} \frac{1}{M} \sum_{\ell=1}^K \frac{\sqrt{c_0'c_0''} \frac{1}{M}\trace \LB \bQ''\bA \bQ_o \RB}{1+\frac{1}{M}\trace \LB \bQ'' \RB}\\ 
&\stackrel{(d)}{\asymp} \frac{1}{\beta} \frac{\sqrt{c_0'c_0''} \delta^2\frac{1}{M}\trace \LB\bA \RB}{1+\delta}.
\label{eq:proof_lemma1_5}
\end{align} 
where equality~$(a)$ follows from Lemma~\ref{lemma_resolvent}, equality~$(b)$ from Lemma~\ref{lemma_trace}, equality~$(c)$ from Lemma~\ref{lemma_rank1}, and equality~$(d)$ from the fundamental Theorem \ref{fundamental_theorem}. The following calculations are very similar and the same lemmas are used in the same way such that we will omit to mention explicitly the lemmas used.
  
Turning to $Z_3$ gives
\begin{align}
Z_3&=-\frac{1}{M^3} \trace \LB\bA \bQ_o \bH'^{\He}\bH' \bQ'\bH'^{\He}\bH''\bQ''\RB\\
&=-\frac{1}{M^3} \sum_{\ell=1}^K \trace \LB\bh_{\ell}'^{\He} \bQ'\bH'^{\He}\bH''\bQ''\bA \bQ_o \bh_{\ell}'\RB\\
&=-\frac{1}{M^3} \sum_{\ell=1}^K \frac{\trace \LB\bh_{\ell}'^{\He} \bQ_{\ell}'\bH'^{\He}\bH''\bQ''\bA \bQ_o \bh_{\ell}'\RB}{1+\frac{1}{M}\bh_{\ell}'^{\He} \bQ_{\ell}'\bh_{\ell}'}\\
&\stackrel{(a)}{=}-\frac{1}{M^3} \sum_{\ell=1}^K \frac{\trace \LB\bh_{\ell}'^{\He} \bQ_{\ell}'\bH'^{\He}\bH''\bQ_{\ell}''\bA \bQ_o \bh_{\ell}'\RB}{1+\frac{1}{M}\bh_{\ell}'^{\He} \bQ_{\ell}'\bh_{\ell}'}+\frac{1}{M^4} \sum_{\ell=1}^K \frac{\trace \LB\bh_{\ell}'^{\He} \bQ_{\ell}'\bH'^{\He}\bH''\bQ_{\ell}''\bh_{\ell}''\bh_{\ell}''^{\He}\bQ_{\ell}''\bA \bQ_o \bh_{\ell}'\RB}{\LB 1+\frac{1}{M}\bh_{\ell}'^{\He} \bQ_{\ell}'\bh_{\ell}'\RB\LB 1+\frac{1}{M}\bh_{\ell}''^{\He} \bQ_{\ell}''\bh_{\ell}''\RB}\\
&\triangleq Z_4+Z_5
\label{eq:proof_lemma1_6}
\end{align}
with equality~$(a)$ obtained using Lemma~\ref{lemma_resolvent}. We also split the calculation in two and start by calculating $Z_4$ as follows.
\begin{align}
Z_4&=-\frac{1}{M^3} \sum_{\ell=1}^K \frac{\trace \LB\bh_{\ell}'^{\He} \bQ_{\ell}'\bH_{\ell}'^{\He}\bH_{\ell}''\bQ_{\ell}''\bA \bQ_o \bh_{\ell}'\RB}{1+\frac{1}{M}\bh_{\ell}'^{\He} \bQ_{\ell}'\bh_{\ell}'}-\frac{1}{M^3} \sum_{\ell=1}^K \frac{\trace \LB\bh_{\ell}'^{\He} \bQ_{\ell}'\bh_{\ell}'\bh_{\ell}''^{\He}\bQ_{\ell}''\bA \bQ_o \bh_{\ell}'\RB}{1+\frac{1}{M}\bh_{\ell}'^{\He} \bQ_{\ell}'\bh_{\ell}'}\\
&=-\frac{1}{M^3} \sum_{\ell=1}^K \frac{\trace \LB\bQ_{\ell}'\bH_{\ell}'^{\He}\bH_{\ell}''\bQ_{\ell}''\bA \bQ_o\RB}{1+\frac{1}{M}\trace \LB\bQ_{\ell}'\RB}-\frac{1}{M} \sum_{\ell=1}^K \sqrt{c_0'c_0''}\frac{\frac{1}{M}\trace \LB \bQ_{\ell}'\RB \frac{1}{M} \trace \LB\bQ_{\ell}''\bA \bQ_o\RB}{1+\frac{1}{M}\trace \LB\bQ_{\ell}'\RB}\\
&\asymp -\frac{K}{M}  \frac{\frac{1}{M^2}\trace \LB\bQ' \bH'^{\He}\bH''\bQ''\bA \bQ_o\RB}{1+\frac{1}{M}\trace \LB\bQ'\RB}-\frac{K}{M} \sqrt{c_0'c_0''}\frac{\frac{1}{M}\trace \LB \bQ'\RB \frac{1}{M} \trace \LB\bQ''\bA \bQ_o\RB}{1+\frac{1}{M}\trace \LB\bQ'\RB}\\ 
&\asymp 	-Z_2-\frac{\delta \sqrt{c_0'c_0''}}{\beta} \frac{\delta^2\frac{1}{M}\trace \LB\bA\RB}{1+\delta}.
\label{eq:proof_lemma1_6}
\end{align}
Finally, it remains to calculate $Z_5$ as  
\begin{align}
Z_5&\asymp \frac{1}{M^4} \sum_{\ell=1}^K \frac{\trace \LB\bh_{\ell}'^{\He} \bQ_{\ell}'\bH'^{\He}\bH''\bQ_{\ell}''\bh_{\ell}''\bh_{\ell}''^{\He}\bQ_{\ell}''\bA \bQ_o \bh_{\ell}'\RB}{\LB 1+ \delta\RB^2}\\
&\asymp \frac{1}{M^4} \sum_{\ell=1}^K \frac{\trace \LB\bh_{\ell}'^{\He} \bQ_{\ell}'\bH_{\ell}'^{\He}\bH_{\ell}''\bQ_{\ell}''\bh_{\ell}''\bh_{\ell}''^{\He}\bQ_{\ell}''\bA \bQ_o \bh_{\ell}'\RB}{\LB 1+ \delta\RB^2}\notag\\
&~~~~~~~~~~~~~~~~~~~~~~~~~~~~~~+\frac{1}{M^4} \sum_{\ell=1}^K \frac{\trace \LB\bh_{\ell}'^{\He} \bQ_{\ell}'\bh_{\ell}'^{\He}\bh_{\ell}''\bQ_{\ell}''\bh_{\ell}''\bh_{\ell}''^{\He}\bQ_{\ell}''\bA \bQ_o \bh_{\ell}'\RB}{\LB 1+ \delta\RB^2}\\
&\asymp \frac{1}{M} \sum_{\ell=1}^K c_0'c_0''\frac{\frac{1}{M^2}\trace \LB \bQ_{\ell}'\bH_{\ell}'^{\He}\bH_{\ell}''\bQ_{\ell}''\RB \frac{1}{M} \trace \LB \bQ_{\ell}''\bA \bQ_o\RB}{\LB 1+ \delta\RB^2}\notag\\
& ~~~~~~~~~~~~~~~~~~~~~~~~~~~~~~+\frac{1}{M} \sum_{\ell=1}^K\sqrt{c_0'c_0''}\frac{\frac{1}{M}	\trace \LB \bQ_{\ell}'\RB \frac{1}{M}	\trace \LB \bQ_{\ell}''\RB \frac{1}{M} \trace \LB \bQ_{\ell}''\bA \bQ_o\RB}{\LB 1+ \delta\RB^2}\\
&\asymp \frac{1}{M} \sum_{\ell=1}^K c_0'c_0''\frac{\frac{1}{M^2}\trace \LB \bQ_{\ell}'\bH_{\ell}'^{\He}\bH_{\ell}''\bQ_{\ell}''\RB \frac{1}{M} \trace \LB \bQ_o\bA \bQ_o\RB}{\LB 1+ \delta\RB^2} +\frac{K}{M} \sqrt{c_0'c_0''}\frac{  \delta^2 \frac{1}{M} \trace \LB \bQ_o\bA \bQ_o\RB}{\LB 1+ \delta\RB^2}\\
&\asymp \frac{1}{\beta}  c_0'c_0''\frac{\frac{1}{M^2}\trace \LB \bQ_{\ell}'\bH_{\ell}'^{\He}\bH_{\ell}''\bQ_{\ell}''\RB  \delta^2\frac{1}{M} \trace \LB\bA\RB}{\LB 1+ \delta\RB^2} +\frac{1}{\beta} \frac{ \sqrt{c_0'c_0''} \delta^4 \frac{1}{M} \trace \LB \bA \RB}{\LB 1+ \delta\RB^2}.
\label{eq:proof_lemma1_8}
\end{align}
Adding all the $Z_i$ gives
\begin{align}
&\frac{1}{M^2} \trace \LB\bA \bQ' \bH'^{\He}\bH''\bQ''\RB\\
&\asymp\LB \frac{1}{\beta} \frac{\sqrt{c_0'c_0''} \delta^2}{1+\delta}-\frac{\delta \sqrt{c_0'c_0''}}{\beta} \frac{\delta^2}{1+\delta}+  \frac{c_0'c_0''\frac{1}{M^2}\trace \LB \bQ_{\ell}'\bH_{\ell}'^{\He}\bH_{\ell}''\bQ_{\ell}''\RB  \delta^2}{\beta\LB 1+ \delta\RB^2}+ \frac{\sqrt{c_0'c_0''}  \delta^4}{\beta\LB 1+ \delta\RB^2}\RB \frac{\trace \LB \bA \RB}{M} \\
&\asymp \frac{\sqrt{c_0c_0'}\delta^2\frac{1}{M} \trace \LB \bA \RB}{\beta (1+\delta)}  \LSB  (1-\delta) + \frac{  \delta^2 \LB 1+\sqrt{c_0'c_0''}\frac{1}{M^2}\trace \LB \bQ_{\ell}'\bH_{\ell}'^{\He}\bH_{\ell}''\bQ_{\ell}''\RB \RB  }{\LB 1+ \delta\RB}\RSB
\label{eq:proof_lemma1_9}
\end{align}
It remains then to calculate the case~$\bA=\bI_{K}$ to conclude the calculation. In that case, we have
\begin{align}
 \LB 1- \frac{c_0'c_0''\delta^2}{\beta\LB 1+ \delta\RB^2}   \RB \frac{1}{M^2} \trace \LB\bQ' \bH'^{\He}\bH''\bQ''\RB&\asymp\frac{\sqrt{c_0'c_0''}\delta^2}{\beta (1+\delta)}  \LSB  (1-\delta) + \frac{  \delta^2  }{\LB 1+ \delta\RB}\RSB
\label{eq:proof_lemma1_10}
\end{align}
%%%%%%%%%%%%%%&=\frac{1}{\beta} \frac{\sqrt{c_0'c_0''} \delta^2}{1+\delta}-\frac{\delta \sqrt{c_0'c_0''}}{\beta} \frac{\delta^2}{1+\delta}+\frac{1}{\beta} \frac{\sqrt{c_0'c_0''}  \delta^2 \frac{1}{M} \trace \LB \bT' \bT\RB}{\LB 1+ \delta\RB^2}\\
Hence,	
\begin{align}
\frac{1}{M^2} \trace \LB\bQ' \bH'^{\He}\bH''\bQ''\RB&\asymp \frac {\frac{ \sqrt{c_0'c_0''}\delta^2}{\beta (1+\delta)}  \LSB  (1-\delta) + \frac{  \delta^2  }{\LB 1+ \delta\RB}\RSB}{ \LB 1- \frac{c_0'c_0''\delta^2}{\beta\LB 1+ \delta\RB^2}   \RB }=Y_o.
\label{eq:proof_lemma1_11}
\end{align}
Inserting \eqref{eq:proof_lemma1_11} inside \eqref{eq:proof_lemma1_10} concludes the proof.
\end{proof}

\begin{lemma}
\label{lemma2}
Let $\bL,\bR,\bar{\bA}\in \mathbb{C}^{M \times M}$ be of uniformly bounded spectral norm with respect to $M$ and let $\bar{\bA}$ be invertible. Let $\xv,\yv$ have i.i.d. complex entries of zero mean, finite variance and finite $8$th order moment and be mutually independent as well as independent of $\bL,\bR,\bar{\bA}$. Then we have:
\begin{align*}
\frac{\xv^{\He}\bL\bA^{-1}\bR\xv}{M} &\asymp u_{\mathrm{LR}}-c_0u_{\mathrm{L}}u_{\mathrm{R}}\frac{1+c_1u}{1+u} +c_2^2u_{\mathrm{L}}u_{\mathrm{R}}\frac{u}{1+u}\\
\frac{\xv^{\He}\bL\bA^{-1}\bR\yv}{M} &\asymp u_{\mathrm{LR}}+c_1c_2u_{\mathrm{L}}u_{\mathrm{R}}\frac{u}{1+u} -c_2 u_{\mathrm{L}}u_{\mathrm{R}}\frac{1+c_1 u}{1+u}
\end{align*}
with
\begin{align*}
\bA&=\bar{\bA}+c_0\xv\xv^{\He}+c_1\yv\yv^{\He}+c_2\xv\yv^{\He}+c_2\yv\xv^{\He}
\end{align*}
with~$c_0+c_1=1$ and $c_0c_1-c_2^2=0$, and 
\begin{equation*}
u \triangleq\frac{\trace( \bar{\bA}^{-1})}{M},\quad u_{\mathrm{L}}\triangleq\frac{\trace(\bL \bar{\bA}^{-1})}{M},\quad u_{\mathrm{R}}\triangleq\frac{\trace(\bar{\bA}^{-1}\bR )}{M},\quad u_{\mathrm{LR}}\triangleq\frac{\trace(\bL\bar{\bA}^{-1}\bR )}{M}.
\end{equation*} 
\end{lemma} 

\begin{proof}
Focusing first on the first equality gives
\begin{align}
&\frac{1}{M}\xv^{\He}\bL\bA^{-1}\bR\xv-\frac{1}{M}\xv^{\He}\bL\bar{\bA}^{-1}\bR\xv\\
&=\frac{1}{M}\xv^{\He}\bL\bA^{-1}\LB\bar{\bA}-\bA \RB   \bar{\bA}^{-1}\bR\xv \\
&=-\frac{1}{M^2}\xv^{\He}\bL\bA^{-1}\LB c_0 \xv \xv^{\He}+c_1\yv \yv^{\He}+c_2 \yv \xv^{\He}+c_2\xv \yv^{\He} \RB   \bar{\bA}^{-1}\bR\xv \\
&\stackrel{(a)}{\asymp}- \frac{1}{M}\LB  c_0\xv^{\He}\bL\bA^{-1}\xv +c_2 \xv^{\He}\bL\bA^{-1}\yv \RB\frac{\trace \LB \bar{\bA}^{-1}\bR\RB}{M} \\
&\stackrel{(b)}{\asymp}- c_0 \frac{\trace \LB \bL\bar{\bA}^{-1}\RB}{M} \frac{\trace \LB \bar{\bA}^{-1}\bR\RB}{M} \frac{1+c_1\frac{\trace \LB \bar{\bA}^{-1}\RB}{M}}{1+\frac{\trace \LB \bar{\bA}^{-1}\RB}{M}}  +c_2^2 \frac{\trace \LB \bL\bar{\bA}^{-1}\RB}{M} \frac{\trace \LB \bar{\bA}^{-1}\bR\RB}{M}\frac{\frac{\trace \LB \bar{\bA}^{-1}\RB}{M}}{1+\frac{\trace \LB \bar{\bA}^{-1}\RB}{M}}
\end{align}
where equality $(a)$ is obtained from using Lemma~\ref{lemma_zero} and Lemma~\ref{lemma_trace} and equality~$(b)$ follows from Lemma~\ref{lemma_c_0}.
Similarly, we turn to the second equality to write
\begin{align}
&\frac{1}{M}\xv^{\He}\bL\bA^{-1}\bR\yv-\frac{1}{M}\xv^{\He}\bL\bar{\bA}^{-1}\bR\yv\\
=&\frac{1}{M}\xv^{\He}\bL\bA^{-1}\LB\bar{\bA}-\bA \RB   \bar{\bA}^{-1}\bR\yv \\
=&-\frac{1}{M^2}\xv^{\He}\bL\bA^{-1}\LB c_0 \xv \xv^{\He}+c_1\yv \yv^{\He}+c_2 \yv \xv^{\He}+c_2\xv \yv^{\He} \RB   \bar{\bA}^{-1}\bR\yv \\
\stackrel{(a)}{\asymp}&-\frac{1}{M}\LB c_1\xv^{\He}\bL\bA^{-1}\yv  +c_2\xv^{\He}\bL\bA^{-1}\xv   \RB   \frac{\trace\LB \bar{\bA}^{-1}\bR\RB}{M} \\
\stackrel{(b)}{\asymp}& c_1  c_2 \frac{\trace \LB \bL\bar{\bA}^{-1}\RB}{M} \frac{\trace \LB \bar{\bA}^{-1}\bR\RB}{M} \frac{\frac{\trace \LB \bar{\bA}^{-1}\RB}{M}}{1+\frac{\trace \LB \bar{\bA}^{-1}\RB}{M}}  -c_2     \frac{\trace \LB \bL\bar{\bA}^{-1}\RB}{M} \frac{\trace \LB \bar{\bA}^{-1}\bR\RB}{M}\frac{1+c_1\frac{\trace \LB \bar{\bA}^{-1}\RB}{M}}{1+\frac{\trace \LB \bar{\bA}^{-1}\RB}{M}}
\end{align}
where equality $(a)$ is obtained from using Lemma~\ref{lemma_zero} and Lemma~\ref{lemma_trace} and equality~$(b)$ follows from Lemma~\ref{lemma_c_0}.
\end{proof}

%%%%%%%%%%%%%%%%%%%%%%%%%%%%%%%%%%%%%%%%%%%%%%%%%%%%%%%%%%%%%%%%%%%%%%%%%%%%%%%%%%%%%%%%%%%%%%%%%%%%%%%%%%%%%%%%%%%%%%
%%%%%%%%%%%%%%%%%%%%%%%%%%%%%%%%%%%%%%%%%%%%%%%%%%%%%%%%%%%%%%%%%%%%%%%%%%%%%%%%%%%%%%%%%%%%%%%%%%%%%%%%%%%%%%%%%%%%%%
%%%%%%%%%%%%%%%%%%%%%%%%%%%%%%%%%%%%%%%%%%%%%%%%%%%%%%%%%%%%%%%%%%%%%%%%%%%%%%%%%%%%%%%%%%%%%%%%%%%%%%%%%%%%%%%%%%%%%%
%%%%%%%%%%%%%%%%%%%%%%%%%%%%%%%%%%%%%%%%%%%%%%%%%%%%%%%%%%%%%%%%%%%%%%%%%%%%%%%%%%%%%%%%%%%%%%%%%%%%%%%%%%%%%%%%%%%%%%
%%%%%%%%%%%%%%%%%%%%%%%%% END SECOND LEMMA     %%%%%%%%%%%%%%%%%%%%%%%%%%%%%%%%%%%%%%%%%%%%%%%%%%%%%%%%%%%%%%%%%%%%%%%
%%%%%%%%%%%%%%%%%%%%%%%%%%%%%%%%%%%%%%%%%%%%%%%%%%%%%%%%%%%%%%%%%%%%%%%%%%%%%%%%%%%%%%%%%%%%%%%%%%%%%%%%%%%%%%%%%%%%%%
%%%%%%%%%%%%%%%%%%%%%%%%%%%%%%%%%%%%%%%%%%%%%%%%%%%%%%%%%%%%%%%%%%%%%%%%%%%%%%%%%%%%%%%%%%%%%%%%%%%%%%%%%%%%%%%%%%%%%%
%%%%%%%%%%%%%%%%%%%%%%%%%%%%%%%%%%%%%%%%%%%%%%%%%%%%%%%%%%%%%%%%%%%%%%%%%%%%%%%%%%%%%%%%%%%%%%%%%%%%%%%%%%%%%%%%%%%%%%
%%%%%%%%%%%%%%%%%%%%%%%%%%%%%%%%%%%%%%%%%%%%%%%%%%%%%%%%%%%%%%%%%%%%%%%%%%%%%%%%%%%%%%%%%%%%%%%%%%%%%%%%%%%%%%%%%%%%%%

\bibliographystyle{IEEEtran}
%%%\bibliography{./../Literature}
\bibliography{Literature}
\end{document}